\documentclass[12pt]{article}
\usepackage{geometry}
\usepackage[
    colorlinks=true,
    linkcolor=blue,
    citecolor=blue,
    urlcolor=blue
]{hyperref}

\usepackage{setspace}
\usepackage{graphicx} 
\graphicspath{ {./figures/} }
\usepackage{amsmath,amssymb} 
\usepackage{natbib}
\usepackage{subcaption}
\usepackage{amsthm}

\usepackage[dvipsnames]{xcolor}
\usepackage{xcolor}
\newcommand{\com}[1]{#1}

\title{Robust Regression with Student's $T$: The Role of Degrees of Freedom}
\usepackage{authblk}

\author[1]{Amanda Ng}
\author[2]{Shangkai Zhu}
\author[3]{Archer Gong Zhang}
\author[1]{Nancy Reid}

\affil[1]{Department of Statistical Sciences, University of Toronto, ON, Canada}
\affil[2]{Department of Statistics, Purdue University, IN, United States}
\affil[3]{School of Mathematics and Statistics, University of Glasgow, Scotland, United Kingdom}

\date{}
\begin{document}
\maketitle

\begin{abstract}
Linear regression estimators are known to be sensitive to outliers, and one alternative to obtain a robust and efficient estimator of the regression parameter is to model the error with Student's $t$ distribution.  In this article, we compare estimators of the degrees of freedom parameter in the $t$ distribution using frequentist and Bayesian \com{maximum a posteriori} methods, and study properties of the corresponding estimated regression coefficients. We include a comparison with some recommended approaches in the literature, including fixing the degrees of freedom and robust regression using the Huber loss function. 
\com{Our results suggest that estimating the degrees of freedom by maximizing the adjusted profile log-likelihood function often yields accurate estimates of the regression coefficients, with performance comparable to that obtained by fixing the degrees of freedom at their true values.}
These findings provide a detailed synthesis of $t$-based robust regression and underscore a key insight: the proper calibration of the degrees of freedom is as crucial as the choice of the robust distribution itself for achieving optimal performance. The \texttt{R} package that implements our method is available at \href{https://github.com/amanda-ng518/RobustTRegression}{RobustTRegression}.
\end{abstract}

\com{heavy-tailed errors},
Jeffreys prior, 
\com{maximum likelihood estimation},
outliers

\section{Introduction}\label{intro}

In this paper, we consider robust estimation in the linear regression model
\begin{equation}\label{eq:linreg}
    y = X\beta + \sigma \epsilon,
\end{equation}
where $y=(y_1, \ldots, y_n)$ is a vector of responses, $X$ is an $n\times p$ matrix of covariates with the $i$th row $x_i^T$, and $\epsilon = (\epsilon_1, \ldots, \epsilon_n)$ is a vector of random errors, independent of $X$. With a view to finding robust estimators of $\beta$, and following \citet{Zellner1976}, \citet{singh1988}, \citet{Lange1989}, and \citet{Fernandez1999}, we assume $\epsilon_i, i = 1, \ldots, n$, are independent and identically distributed from standard Student's $t$ distribution with degrees of freedom $\nu$. 

Robustness is a fundamental concern in statistical inference. The presence of outliers can distort parameter estimates and degrade the reliability of inference. Other methods developed for robust inference include  $M$-estimation, where the estimating function is designed to down-weight outliers  \citep{huber1973robust}, $L$-estimation, based on linear combinations of order statistics \citep{koenker1987}, and $R$-estimation, based on ranks \citep[Chapter 3]{huber_2009}. 

\com{Robustness in regression models can also refer to the effects of outliers, influential observations, and multicollinearity in the design matrix. A robust estimator for linear regression is proposed in \citet{dawoud2021robust}, and extended to some generalized linear models in \citet{dawoud2022new} and \citet{dawoud2025new}, for example. Influence diagnostics for regression models in the presence of multicollinearity are developed in \citet{dawoud2025detection}. This work primarily focuses on improved estimators, whereas we consider instead a more robust choice for the error distribution.}

For model \eqref{eq:linreg} with Student's $t$ errors, the likelihood function for $(\beta, \sigma, \nu)$  is 
\begin{equation}
\label{eq:student_t_likelihood}
    L(\beta, \sigma, \nu; y, X) =  \frac{\Gamma^n((\nu+1)/2)}{\sigma^n (\pi \nu)^{n/2} \Gamma^n(\nu/2)} \prod_{i=1}^n \left\{1 + \frac{(y_i - x_i^{T}\beta)^2}{\nu \sigma^2} \right\}^{-(\nu+1)/2}. 
\end{equation}
Whether the covariate matrix $X$ is fixed or random, inference is based on the conditional density of $y$ given $X$, as is usual in regression problems. The estimates of $\beta$ and $\sigma$ derived from \eqref{eq:student_t_likelihood} are expected to be more robust than least squares estimates, as the $t$ distribution has heavier tails than the normal distribution, especially for smaller values of $\nu$. As the degrees of freedom $\nu\rightarrow\infty$, the corresponding Student's $t$ distribution converges to the normal distribution.

The ordinary least squares estimator of $\beta$ is most efficient under the normal model, but it is sensitive to outliers. The maximum likelihood estimator of $\beta$ from \eqref{eq:student_t_likelihood}  with smaller values of $\nu$ can be more robust, but may be inefficient if the data are consistent with a larger value of $\nu$. Thus, the choice of $\nu$ determines the trade-off between robustness and efficiency. \com{
\citet{Lange1989} developed a general maximum likelihood framework for models with Student's $t$ errors and applied it to both univariate and multivariate responses. They suggested estimating $\nu$ when $n$ is large but using a fixed value when $n$ is small;  joint estimation of $\nu$ with the regression and scale parameters can be difficult in small samples. They used the expectation-maximization (EM) algorithm to maximize \eqref{eq:student_t_likelihood} over all parameters. \citet{liu1995} showed that ECME, a generalization of the EM algorithm, can reduce the computational effort in obtaining the maximum likelihood estimate of $(\beta,\sigma,\nu)$.}

\com{Although maximizing the likelihood provides a natural approach to estimate $\nu$, the resulting estimator can exhibit substantial finite-sample bias, especially when the sample size is small. \citet{Vasconcellos2005} proposed a bias-corrected maximum likelihood estimator for $\nu>1$.
Bayesian methods provide another framework to estimate $\nu$, although the choice of prior distribution is nontrivial. Early work by \citet{Relles1977} considered a simple uniform prior. Noninformative Jeffreys-type priors were proposed and studied by \citet{Fonseca2008, Fernandez1999}, and \citet{he2021objective} investigated reference priors.}

In this paper, we suggest using the adjusted profile likelihood estimator of $\nu$. \com{This starts with the profile likelihood function $\ell_{\rm p}(\nu)$, defined in \eqref{eq:logprofile}, and then adds an information adjustment to it to account for the estimation of the regression coefficients and scale parameters. This differs from the bias-correction approach of \citet{Vasconcellos2005}, which directly corrects the leading $O(n^{-1})$ term in the bias of the maximum likelihood estimator for small sample sizes. In contrast, the adjusted profile likelihood estimator modifies the inference function, which aims to reduce the bias in the score equation that defines the estimator.}

In Section 2, we compare likelihood and Bayesian \com{maximum a posteriori} approaches to the estimation of $\nu$. In Section 3, we focus on inference for $\beta$, where we first estimate the degrees of freedom parameter $\nu$, and then fix $\nu$ at this estimated value in the likelihood function to maximize over $\beta$ and $\sigma$. For moderately large values of $p$, we find that estimating $\nu$ by maximizing the adjusted profile likelihood function gives \com{average normalized $L_2$ estimation error} for $\beta$ comparable to that obtained by minimizing Huber's loss function \citep{huber1973robust}. In Section 4, we consider the robustness of our proposed method in the presence of outliers. 
Implementations of the methods in the \texttt{R} \citep{R-4.3.0} package \texttt{RobustTRegression} are available on GitHub at \href{https://github.com/amanda-ng518/RobustTRegression}{https://github.com/amanda-ng518/RobustTRegression}.


\section{Estimating the degrees of freedom}\label{est_nu}

\subsection{Likelihood-based approaches}
The first estimator of $\nu$ we consider is the maximum likelihood estimator, $\hat{\nu}$, which maximizes the  profile log-likelihood function 
\begin{equation}\label{eq:logprofile}
\ell_{\rm p}(\nu; y, X) = \ell(\hat{\beta}_\nu, \hat{\sigma}_\nu, \nu; y, X),
\end{equation}
where $\ell(\theta) = \log L(\theta;y, X)$, and $(\hat \beta_\nu, \hat \sigma_\nu) = \arg\max_{(\beta, \sigma)}\ell(\beta, \sigma, \nu; y, X)$ are the constrained maximum likelihood
estimators of $(\beta, \sigma)$. \citet{maronna1976} and \citet{Zellner1976} showed that $\hat \beta_\nu$ and $\hat\sigma_\nu$ exist.
The asymptotic properties of $\hat{\nu} $ follow from classical likelihood theory:
$$
j_{\rm p}^{1/2}(\hat\nu) (\hat\nu - \nu) \xrightarrow{d} N(0, 1), \quad n \rightarrow\infty,
$$
where $j_{\rm p}(\hat\nu) = -\ell_{\rm p}''(\hat\nu;y, X)$ can be computed from the Fisher information matrix for $\theta = (\beta,\sigma, \nu)$.

Since the parameter of interest $\nu$ is a scalar and $\lambda=(\beta,\sigma)$ is a nuisance parameter, we consider the adjusted profile log-likelihood function in \citet{cox1987},
\begin{equation}\label{eq:logadj}
\ell_{\mathrm{adj}}(\nu;y)
=
\ell_{\mathrm{p}}(\nu;y)
-
\frac{1}{2}
\log
\left|
j_{\lambda\lambda}
\left(\nu,\widehat{\lambda}_{\nu}\right)
\right|,
\end{equation}
where $j_{\lambda\lambda}(\nu, \widehat\lambda_\nu)$ is the submatrix of the observed Fisher information matrix that corresponds to the nuisance parameters. An expression for $j_{\lambda\lambda}(\nu,\lambda)$ is given in the supplementary material. \com{This adjustment term corrects the profile likelihood for errors of estimation of the nuisance parameters, at least approximately. It is one of a class of modified profile likelihood functions suggested in \citet{BarndorffNielsen:1983}, and the easiest among these to compute. 
\citet{cox1987} recommended that this adjusted profile likelihood function be used when the parameter of interest (here $\nu$) is orthogonal to the nuisance parameters, with respect to the expected Fisher information matrix $I(\theta) = {\rm E}\{-\partial^2\ell(\theta)/\partial\theta\partial\theta^T\}$, meaning the elements of $I(\theta)$ corresponding to $(\nu, \beta)$ and $(\nu, \sigma)$ are $0$. It is convenient to choose a fixed parametrization, as \eqref{eq:logadj} is not parametrization-invariant. Further, when the nuisance parameters  $\beta$ and $\sigma$ are orthogonal to the parameter of interest $\nu$, the adjusted profile log-likelihood is an approximate conditional likelihood for $\nu$. In this Student's $t$ regression model, while $\beta$ is
information-orthogonal to both $\sigma$ and $\nu$,  $\sigma$
is not orthogonal to $\nu$; see \eqref{eq:expfisher} below. 

Nevertheless, we apply the
Cox--Reid determinant adjustment as an approximate nuisance-parameter
adjustment and investigate its finite-sample performance empirically.
The value of $\nu$ that maximizes the adjusted profile log-likelihood function \eqref{eq:logadj} is denoted by $\hat{\nu}_{\rm adj}$.}

\subsection{Priors for the degrees of freedom}

\citet{Jeffreys1946-rl} proposed a class of noninformative priors constructed from the expected Fisher information matrix, which is \citep{Fonseca2008}
\begin{equation}\label{eq:expfisher}
I(\beta, \sigma, \nu) = 
\begin{pmatrix}
\frac{1}{\sigma^2} \frac{\nu + 1}{\nu + 3} \sum\limits_{i=1}^{n} x_i x_i^T & 0 & 0 \\
0 & \frac{2n}{\sigma^2} \frac{\nu}{\nu + 3} & -\frac{2n}{\sigma} \frac{1}{(\nu + 1)(\nu + 3)} \\
0 & -\frac{2n}{\sigma} \frac{1}{(\nu + 1)(\nu + 3)} & 
\frac{n}{4} \left[
\psi'\left( \frac{\nu}{2} \right) -
\psi'\left( \frac{\nu + 1}{2} \right) -
\frac{2(\nu + 5)}{\nu(\nu + 1)(\nu + 3)}
\right]
\end{pmatrix}.
\end{equation}
\com{The Jeffreys-rule prior \citep{Kass.Wasserman:1996}, is proportional to $|I(\beta,\sigma,\nu)|^{1/2}$, but was shown in \citet{Fonseca2014} to lead to an improper posterior. \citet{Kass.Wasserman:1996} noted that for regression models \citet{Jeffreys:1961} recommended the independence prior, in which a flat prior is used for $\beta$, independent of the prior for  $(\sigma, \nu)$. This prior was studied in \citet{Fonseca2008} and shown to lead to a proper posterior when $n\ge p+1$.} This independence Jeffreys prior is 
\begin{align}
\pi_{IJ}(\beta)  \propto 1, \quad \quad
\pi_{IJ}(\sigma, \nu)  \propto \sigma^{-1}
 \left \{\frac{\nu}{\nu+3} \right\}^\frac{1}{2} \left\{ \psi'\left(\frac{\nu}{2}\right) - \psi'\left(\frac{\nu + 1}{2}\right) - \frac{2(\nu + 3)}{\nu (\nu + 1)^2}\right\}^\frac{1}{2}; \label{eq:IJ_prior}
\end{align}
the prior for $(\sigma, \nu)$ is the square-root of the determinant of the $2 \times 2$ submatrix of \eqref{eq:expfisher}.
Using these priors, we maximized the posterior distribution over $(\beta, \sigma, \nu)$; this Bayes estimator of $\nu$ is the maximum a posteriori (MAP) estimator $\hat{\nu}_{\rm Bayes}$. 

An alternative partially Bayesian approach is to treat the profile likelihood function for $\nu$ as a true likelihood function and use the prior $\pi(\nu)\propto |I_{\nu\nu}(\nu)|^{1/2}$ to form a pseudo-posterior distribution. The MAP estimator of $\nu$ obtained from this pseudo-posterior distribution is denoted by $\hat{\nu}_{\rm Pseudo}$. \com{We note that our focus is on point estimation of $\nu$ and its effect on regression coefficient estimation. We use MAP-based estimates because they incorporate the specified priors while retaining an optimization-based formulation comparable to the
likelihood-based approaches. Our Bayesian comparisons are restricted to these MAP-based procedures.}

\com{Estimators of $\nu$ from these objective functions are invariant to location-scale transformations. For the linear model \eqref{eq:linreg}, transforming the response to $y_i' = ay_i + x_i^T b$ with $a>0$ and a constant vector $b$ shifts the profile log-likelihood, adjusted profile log-likelihood, and corresponding log-posterior objectives only by additive constants that do not depend on $\nu$. Then, the maximizers with respect to $\nu$ would be unchanged.}

Following \citet{Lange1989}, we estimated the degrees of freedom $\nu$ using the reparametrization $\omega = 1/\nu$, so $\omega \rightarrow 0$ as $\nu \rightarrow \infty$ and the Student's $t$ distribution converges to the normal distribution. When specifying priors, if $p(\nu)$ is the prior for $\nu$, then the prior for $\omega$ is $\pi(\omega) = p(1/\omega)\,/\omega^2$.
\com{In the numerical comparisons, we maximize the transformed posterior or pseudo-posterior density in $\omega=1/\nu$, including the Jacobian, and report the reciprocal of the resulting maximizer.}

\subsection{Simulations}
To compare the performance of the four estimators of $\nu$ described above, we generated data from the regression model \eqref{eq:linreg} with Student's $t$ errors for $\epsilon_i$, degrees of freedom $\nu \in \{ 2, 5, 10\}$, scale parameter $\sigma = 1$, and $\beta \in \mathbb{R}^p$  with dimension $p \in \{ 1, 2, 5, 10, 20, 40, 60, 80\}$. The covariates were generated from a multivariate standard normal distribution. 
The choice of the true values of $\beta$ and $\sigma$ is arbitrary, as the inference for $\nu$ is invariant to location-scale changes. The parameters we chose for $(n, \nu)$ are $(300, 2), (2500, 5)$ and $(4500, 10)$.

The BFGS algorithm provided in \texttt{R} \citep{R-4.3.0} was used for all optimization steps in our simulations. This is a quasi-Newton method that iteratively approximates the Hessian of the objective function. The starting values of $(\beta, \sigma)$ for constrained maximum likelihood estimation were the ordinary least squares estimate of $\beta$ and the square root of the usual residual mean square.

\com{
The four inference functions were optimized  in the $\omega=1/\nu$ parameterization, with results reported in terms of $\nu$. 
To obtain a data-dependent starting value, we first maximized the log-likelihood function in \eqref{eq:student_t_likelihood} over $\nu$, with $\beta$ and $\sigma$ fixed at the least square estimates.
This preliminary optimization was initialized at $\nu=20$.
Our unreported sensitivity checks using alternative starting values indicated that the resulting preliminary estimate was relatively insensitive to this choice in the settings examined. 
Denoting this preliminary estimate by $\tilde{\nu}$, we used $\omega_{\rm init}=1/\tilde{\nu}$ as the starting value for optimizing each of the four inference functions described in the previous subsections.}
 
We ran $500$ successful simulations for each combination of $\nu$ and $p$, with the covariates fixed across all repetitions. The supplementary material provides more details about the optimization, including starting values, sample sizes, and the methods used to handle intermediate values of $\nu$ that were out-of-bounds.

\begin{table*}[!t]%
\centering %
\caption{RMSE of estimates of $\nu$, with true $\nu = 2$ and sample size $n = 300$.\label{tab:rmse_nu2}}%
\begin{tabular*}{\textwidth}{@{\extracolsep\fill}lllll@{\extracolsep\fill}}
\hline
\textbf{$p$} & \textbf{$\hat{\nu}$} & \textbf{$\hat{\nu}_{\rm adj}$} & \textbf{$\hat{\nu}_{\rm Bayes}$} & \textbf{$\hat{\nu}_{\rm Pseudo}$} \\
\hline
1 & 0.3453 & 0.3506 & 0.3541 & 0.3573 \\
2 & 0.3439 & 0.3521 & 0.3525 & 0.3556 \\
5 & 0.3361 & 0.3529 & 0.3437 & 0.3464 \\
10 & 0.3361 & 0.3652 & 0.3420 & 0.3442 \\
20 & 0.3381 & 0.3845 & 0.3404 & 0.3409 \\
40 & 0.5995 & 0.4495 & 0.5834 & 0.5553 \\
60 & 1.4456 & 0.5838 & 1.4617 & 1.4377 \\
80 & 1.6800 & 0.9053 & 1.6757 & 1.6769 \\
\hline
\end{tabular*}
\end{table*}

Table~\ref{tab:rmse_nu2} presents the root mean square error (RMSE) averaged over the 500 simulations for the four estimators of $\nu$, when the true value of $\nu = 2$ and the sample size $n=300$. For small values of $p$, the profile maximum likelihood estimator $\hat{\nu}$ has the smallest RMSE, but all four methods are comparable. For $p=60$ or $80$, the adjusted profile likelihood estimator
$\hat{\nu}_{\rm adj}$ has the smallest RMSE, while for smaller $p$ its performance is comparable to that of the maximum likelihood estimator $\hat{\nu}$. \com{These results support considering the adjusted profile likelihood estimator as a useful practical option for estimating $\nu$.}

The two \com{Bayesian-type} estimates have quite similar performance to the maximum likelihood estimator and to each other. Other simulation settings of ($\nu$, $n$, $p$), as reported in the supplementary material, give similar conclusions.

\com{ To assess the numerical reliability of these results, we also examined the convergence of the optimization algorithms and the behaviour of their objective functions. 
Although numerical optimization was generally successful, with at most two out of 500 simulation repetitions resulting in optimization errors for the smaller sample size $n=300$, numerical difficulties were observed even when the algorithms reported successful convergence.
These included potential boundary issues in BFGS, numerical instability of the objective functions near their estimated optima, and flatness in settings with larger $p$ relative to $n$.}

\com{
We used the BFGS algorithm without constraints, in the $\omega$-parameterization, setting the objective function to $-\infty$ whenever $\omega<0$. If a line-search trial produced a negative value of $\omega$, the trial was rejected, and the step size was reduced to seek an admissible trial value.
An accepted iterate can lie sufficiently close to $\omega=0$ that the subsequent finite-difference gradient calculation evaluates the maximized objective at a negative value of $\omega$. Since this objective is set to $-\infty$ for $\omega<0$, the resulting gradient approximation is non-finite and \texttt{optim()} in \texttt{R} terminates with an error. An example of an unsuccessful run is shown in Figure~S4b of the supplementary material.
}

\com{
Numerical approximations to the first and second
derivatives of the objective functions with respect
to $\omega$ are reported in the supplementary material.
For $n=300$, $\nu=2$, and $p=60$ or $80$, these
approximations are extremely large for the profile
likelihood and the two Bayesian-type approaches,
but substantially smaller for the adjusted profile
likelihood. 
The corresponding objective functions appear to exhibit numerical instability near their estimated optima, as illustrated in Figure~S6 of the supplementary material for $n=300$, $p=60$, and $\nu=2$.
The former three approaches also produce estimates with $\hat{\omega}>1$ in these settings, corresponding to $\hat{\nu}<1$.

}

\com{
Although outright optimization failures were
infrequent in the main simulation settings, we observed that
numerical difficulties were considerably
more pronounced as $p$ increased for a given value of $n$.
For example, in the supplementary simulations
with $n=100$ and $\nu=2$, optimization failures occurred in approximately 3.4 \% and 66.8\% of repetitions
for $p=40$ and $p=60$, respectively. 
Moreover, when the optimizer reported a successful convergence, this did not necessarily guarantee a reliable estimate of $\nu$ due to relative flatness and numerical instability issues in the objective functions (see Figures~S1, S3 and~S4 of the supplementary material). These difficulties motivated our choice of relatively large sample sizes in the main simulation study, which should therefore be viewed as comparisons under relatively favourable optimization conditions.

}

\section{Estimating the regression coefficient} \label{est_beta}

In linear regression, estimating the regression coefficient $\beta$ is usually the primary objective. To evaluate the impact of estimation of $\nu$  on estimates of $\beta$, we designed a simulation study based on the stack loss data of \citet{brownlee1965}.
This data has 21 responses ($Y$) of an industrial process,  with three explanatory variables: air flow ($X_1$), temperature ($X_2$), and acid concentration ($X_3$). 

\subsection{Simulation settings}

\com{We conducted simulations from the linear model \eqref{eq:linreg} with errors generated from either a standard Student's $t$ distribution with $2$ degrees of freedom or a standard normal distribution.}

  In the first setting, we used the design matrix $X$ of the stackloss dataset, with $n=21$ and $p = 4$ predictors, including an intercept. \com{The true values of $\beta$ and $\sigma$ used for the simulation were the least squares estimates from the original data,  $\hat\beta_{\rm OLS} = (-39.920, 0.716, 1.300, -0.152)$ and $\hat\sigma_{\rm OLS} = 3.243$.}

    In the second set of simulations, we increased the dimension $p$ to $\{ 40,80,120\}$ and the sample size $n$ to $ 210$. For the predictors, we retained the original 21 observations and simulated an additional $189$ rows of the design matrix. \com{The first four columns were generated from a multivariate normal distribution with the same empirical mean vector and covariance matrix as the original $21\times 4$ design matrix.} The remaining $p-4$ columns were generated independently from a standard normal distribution. \com{The values of $y$ generated from \eqref{eq:linreg} used $\mathbf{\beta} =(\hat{\mathbf{\beta}}_{\rm OLS}, \mathbf{0})$ and $\sigma=\hat{\sigma}_{\rm OLS}$.} 
    
    In the third set of simulations, to study the effect of the error distribution, $\epsilon_i$ was simulated either from the $t_2$ or the standard normal distribution. \com{The design matrix was  generated from the $N_p(0, I)$ distribution, and we set $\beta = \mathbf{0}$ and $\sigma = 1$.} Three pairs of $(n,p)$ were used: $(20, 4)$, $(500, 4)$, and $(100, 50)$.

\com{We evaluated the regression estimates using four possible choices for the degrees of freedom, $\nu$: (1) the maximum likelihood estimate $\hat{\nu}$, (2) the adjusted profile likelihood estimate $\hat{\nu}_{\rm adj}$, (3) the MAP estimate, $\hat{\nu}_{\rm Bayes}$, based on Jeffreys' independence prior, and (4) a given fixed value of $\nu$. We discarded the pseudo-Bayesian approach explored in Section~\ref{est_nu} as it gave very similar estimates of $\nu$ and was more difficult to implement. The estimates of $\beta$ and $\sigma$ were obtained by maximizing the Student's $t$ likelihood function with $\nu$ set to one of the three estimates or to its given fixed value.} We also computed the ordinary least squares (OLS) estimator and Huber's robust estimator \citep{huber1973robust} of $\beta$. Huber's method requires a tuning parameter $k$, which is adaptively chosen in the \texttt{R} package \texttt{rlmDataDriven} \citep{rlmDataDrivenRpackage} that we used.  
\com{For a sensitivity analysis, we compared this data-driven
choice with fixed tuning constants $k=1,\ldots,5$ and
$k=1.345$, which is a conventional choice corresponding to approximately 95\% asymptotic efficiency relative to OLS under normal errors \citep{holland1977robust}.}This procedure was repeated over 500 simulations of the responses, again conditioning on the design matrix $X$.

The various loss functions for the optimization of $\beta$ are compared in Figure~\ref{fig:loss-func}. Least squares loss imposes the largest penalty on the outliers compared to other methods. The loss function for $t_1$ imposes a lighter penalty on outliers than Huber's loss function with $k= 1.5$. 

\begin{figure}[ht]
    \centering
    \includegraphics[scale = 0.6]{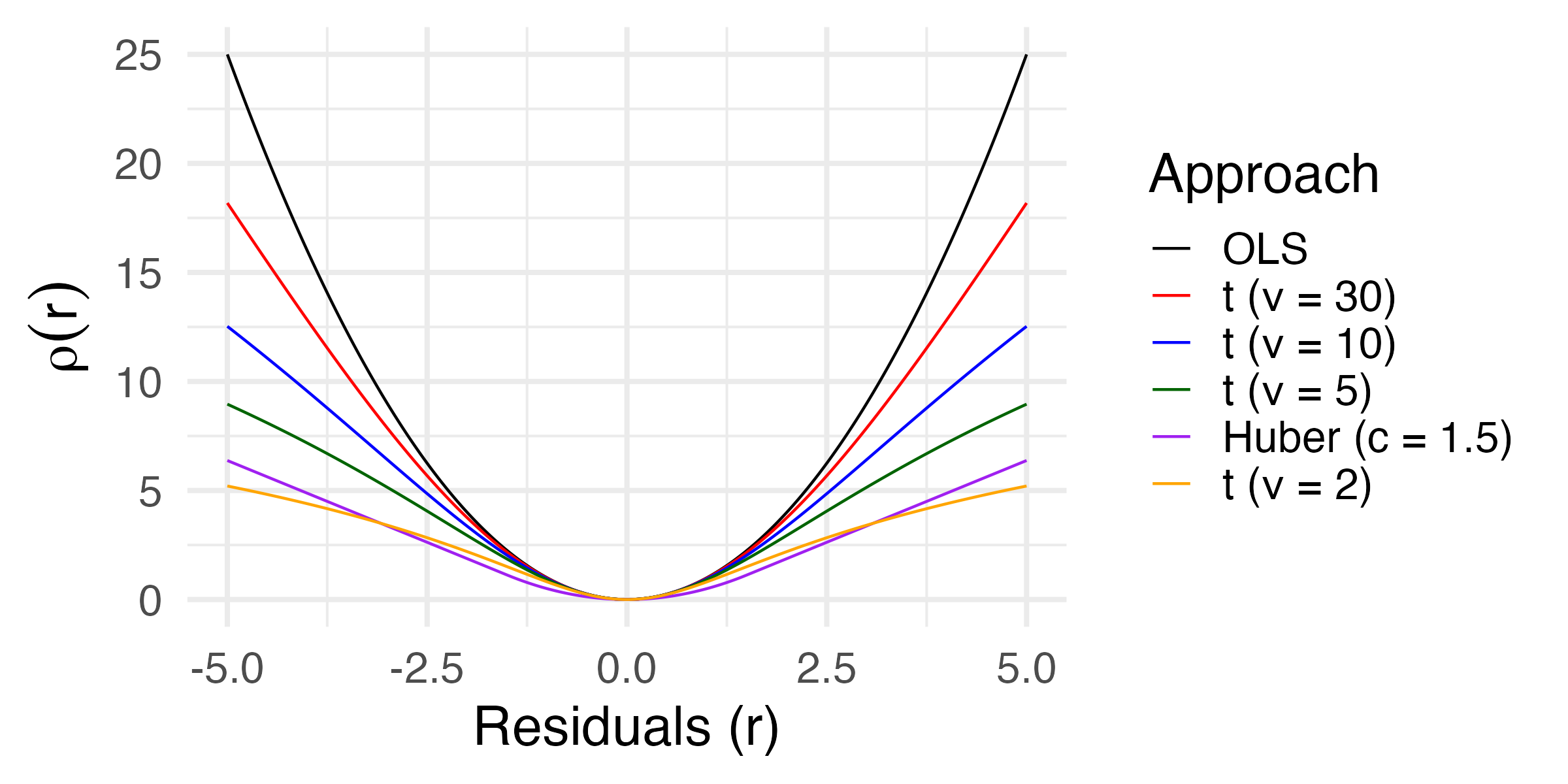}
    \caption{Loss functions for estimation by ordinary least squares (OLS), Huber's method (Huber), and the $t$ distribution for selected values of $\nu$, as functions of the residual $r$.}
    \label{fig:loss-func}
\end{figure}

\com{
Following the use of mean $L_2$ estimation error in robust regression simulation studies, such as \citet{sun2020adaptive}, we use the average normalized $L_2$ estimation error as the performance criterion for comparing regression coefficient estimators. 
Specifically, for each simulation repetition $k=1,\ldots,N_{\mathrm{sim}}$, where we set $N_{\mathrm{sim}}=500$, let
\[
D_k = \sqrt{\frac{1}{p-1}\sum_{j=2}^p
(\hat{\beta}_{jk}-\beta_j)^2},
\]
where $\hat{\beta}_{jk}$ is the estimate of the $j$th component of $\beta$ in the $k$th simulation. Our performance criterion is then
\[
\bar D = \frac{1}{N_{\mathrm{sim}}}\sum_{k=1}^{N_{\mathrm{sim}}}D_k.
\]
The intercept is excluded, and the normalization by $\sqrt{p-1}$
accounts for the varying number of regression coefficients across the simulation settings without affecting the ranking of methods within each setting.
For brevity, we may also refer to this criterion as the average $L_2$ error, or simply the $L_2$ error, hereafter.
}

\com{
This criterion differs from the conventional RMSE, 
\[
\mathrm{RMSE}(\hat{\beta})
= \sqrt{\frac{1}{N_{\mathrm{sim}}}\sum_{k=1}^{N_{\mathrm{sim}}}D_k^2}.
\]
The two criteria satisfy
\[
\mathrm{RMSE}^2(\hat{\beta})
= \bar D^2 + \frac{1}{N_{\mathrm{sim}}}\sum_{k=1}^{N_{\mathrm{sim}}}(D_k-\bar D)^2.
\]
Thus, conventional RMSE places greater weight on occasional large estimation errors. This distinction is particularly
relevant in heavy-tailed settings. 
For example, under Student's $t_2$ errors and a fixed full-rank design, the OLS estimator has finite expected normalized $L_2$ estimation error, whereas its population RMSE is infinite.
Nevertheless, our unreported simulation results using conventional RMSE showed that the two criteria yielded the same rankings of the methods across the simulation settings considered in this paper.
}

\subsection{Results}

\begin{figure}[h!]
    \centering
    \resizebox{1.0\textwidth}{!}{%
    \begin{minipage}{\textwidth}
        \begin{subfigure}[t]{\linewidth}
            \centering
            \includegraphics[width=\linewidth]{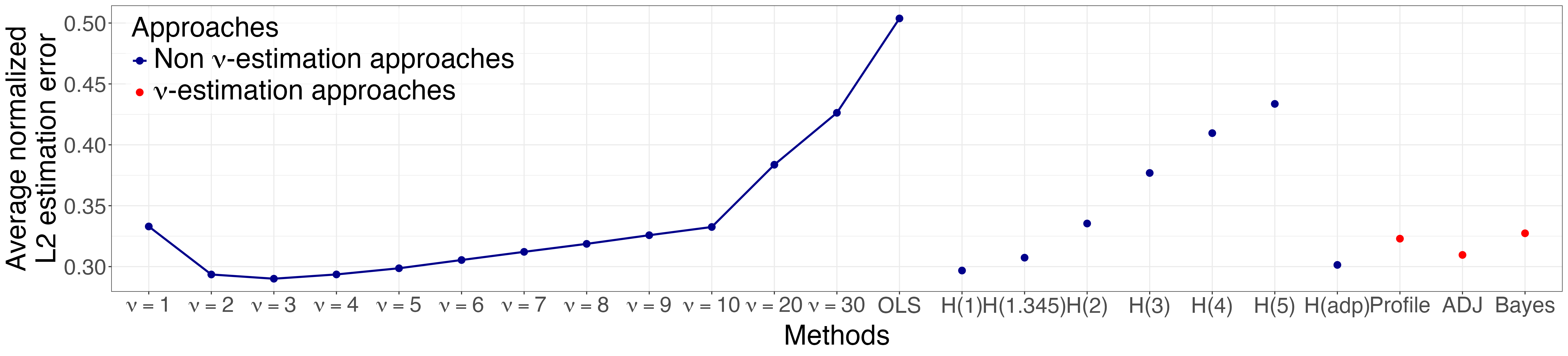}
        \end{subfigure}
    \end{minipage}
    }

    \caption{Estimation of $\hat{\beta}$ with the stackloss data ($n=21, p=4$). Fixed values of $\nu$ and other OLS and Huber estimates are shown in blue. The Huber constant $k$ was either fixed ($H(k)$) or estimated ($H(\rm adp)$). Estimates of $\nu$, either by maximum likelihood, $\hat\nu$, adjusted profile likelihood, $\hat\nu_{\rm adj}$, or MAP with Jeffreys prior, $\hat\nu_{\rm Bayes}$, are shown in red. 
}
    \label{fig:rmse_stackloss_2a}
\end{figure}

\begin{figure}[h!]
    \centering

    \resizebox{0.9\textwidth}{!}{%
    \begin{minipage}{\textwidth}

        \begin{subfigure}[t]{\linewidth}
            \centering
            \includegraphics[width=\linewidth]{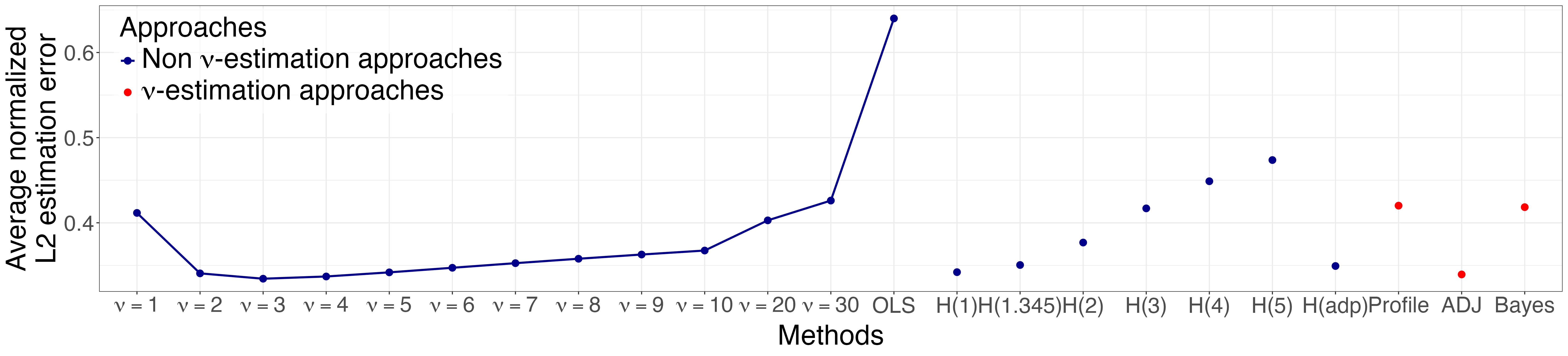}
            \caption{$n=210, p=40$}
            \label{fig:rmse_stackloss_2b}
        \end{subfigure}

        \vspace{1em}

        \begin{subfigure}[t]{\linewidth}
            \centering
            \includegraphics[width=\linewidth]{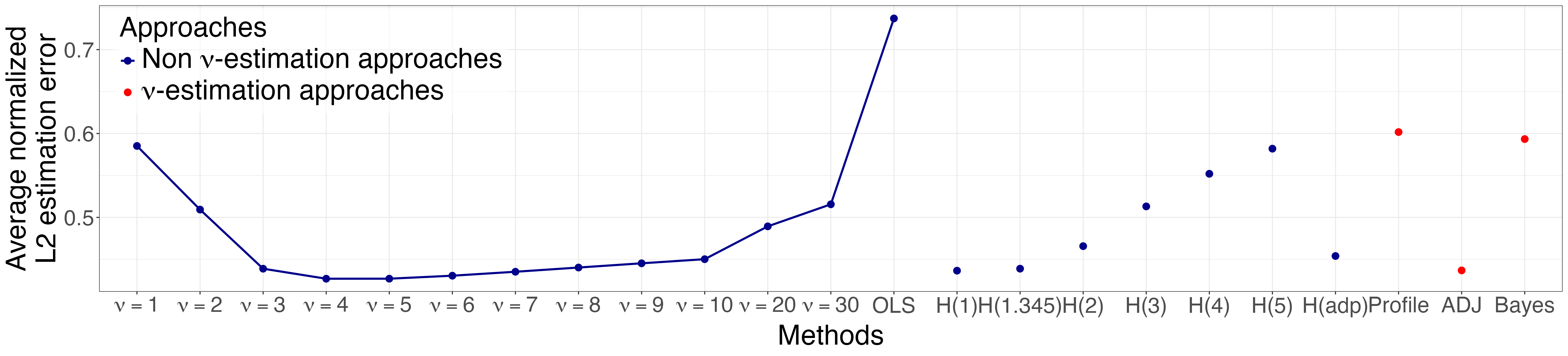}
            \caption{$n=210, p=80$}
            \label{fig:rmse_stackloss_2c}
        \end{subfigure}

        \vspace{1em}

        \begin{subfigure}[t]{\linewidth}
            \centering
            \includegraphics[width=\linewidth]{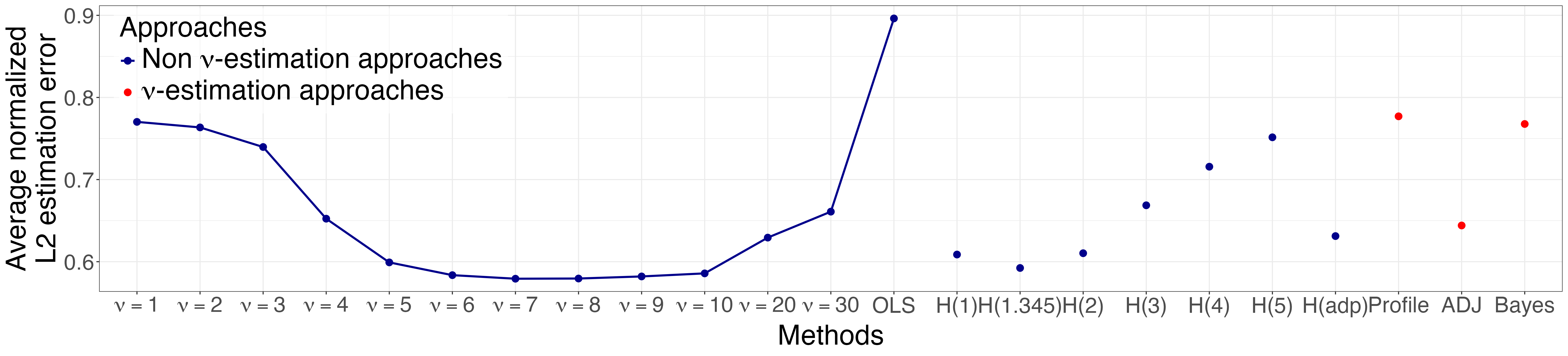}
            \caption{$n=210, p=120$}
            \label{fig:rmse_stackloss_2d}
        \end{subfigure}
    \end{minipage}
    }

    \caption{Estimation of $\hat{\beta}$ using the expanded stackloss data. Fixed values of $\nu$ and other OLS and Huber estimates are shown in blue. The Huber constant $k$ was either fixed ($H(k)$) or estimated ($H({\rm adp})$). Estimates of $\nu$, either by maximum likelihood, $\hat\nu$, adjusted profile likelihood, $\hat\nu_{\rm adj}$, or MAP with Jeffreys prior, $\hat\nu_{\rm Bayes}$, are shown in red.}
    \label{fig:rmse_stackloss_2}
\end{figure}
From Figure~\ref{fig:rmse_stackloss_2a}, we observe that the lowest \com{average $L_2$ error} for $\beta$ is achieved when $\nu$ is fixed at $ 3$, with the \com{average $L_2$ error} at the true value $\nu = 2$ only slightly larger. As $n$ increases, the \com{$L_2$ error} is smallest at the true value of $\nu$, and increases notably with $\nu$. The three $\nu$-estimation approaches give $L_2$ error comparable to the best fixed-$\nu$ results, with $\hat\nu_{\rm adj}$ giving smaller $L_2$ error than $\hat\nu$ or $\hat\nu_{\rm Bayes}$. \com{Huber regression (with $k= 1$ or $ 1.345$, or estimated adaptively) gives slightly smaller $L_2$ error than does estimation of $\beta$ with $\hat\nu_{\rm adj}$.}

Figures~\ref{fig:rmse_stackloss_2b} and~\ref{fig:rmse_stackloss_2c} have $p/n = 0.19, 0.38$ respectively. The smallest value of the \com{average $L_2$} error is achieved by fixing $\nu$ at $3$ or $4$. Using $\hat{\nu}_{\rm adj}$ gives comparable results, slightly better than Huber regression, and comparable to the best fixed value. In Figure~\ref{fig:rmse_stackloss_2d}, with $p/n = 0.57$, the \com{lowest average $L_2$ error} is achieved at larger fixed $\nu = 7$. When $p$ is large, the model has enough flexibility to absorb extreme values into the fitted coefficients, shrinking the residuals relative to the true $t_2$ errors. Using a likelihood function with a larger value of $\nu$ implies a less robust (steeper) loss function (see  Figure~\ref{fig:loss-func}), as the data no longer exhibit the heavy-tailed behaviour that would justify a small $\nu$. Among the $\nu$-estimation approaches, again using $\hat{\nu}_{\rm adj}$ gives the \com{smallest $L_2$ error}, but this value is relatively large
compared to \com{the $L_2$ error} that can be achieved by traditional methods, especially \com{Huber with smaller values and adaptive estimation of $k$}.

We next consider the fully simulated settings to examine the effect of the error distribution. 
With $t_2$ errors, the results broadly follow the patterns observed in the \com{stackloss-based simulations}. We therefore present these results in Figure~S7 of the supplementary material to avoid repeating the comparisons above.

\begin{figure}[h!]
    \centering
    \resizebox{0.9\textwidth}{!}{%
    \begin{minipage}{\textwidth}

        \begin{subfigure}[t]{\linewidth}
            \centering
            \includegraphics[width=\linewidth]{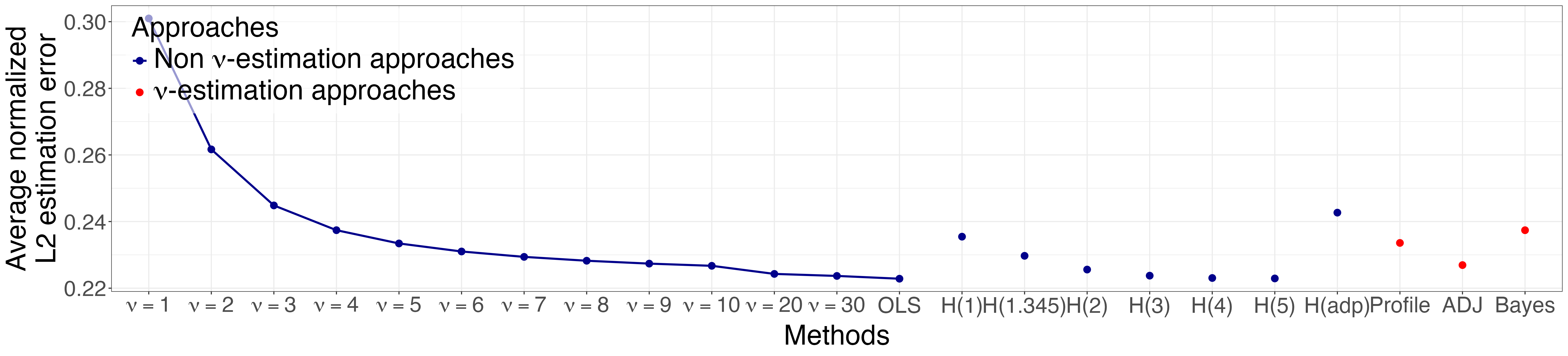}
            \caption{$n=20, p=4$}
            \label{fig:rmse_nerror_a}
        \end{subfigure}
        \vspace{1em}
        \begin{subfigure}[t]{\linewidth}
            \centering
            \includegraphics[width=\linewidth]{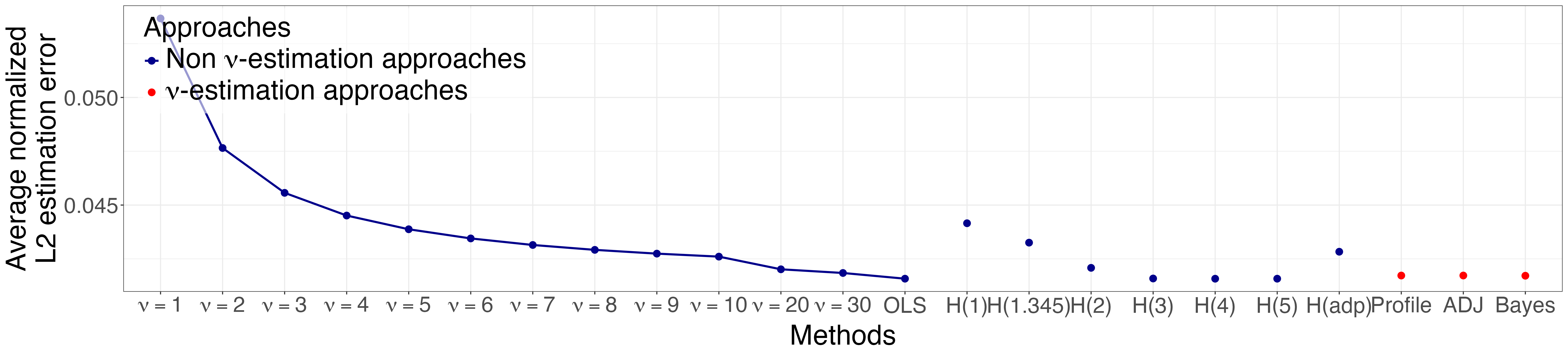}
            \caption{$n=500, p=4$}
            \label{fig:rmse_nerror_b}
        \end{subfigure}
        \vspace{1em}
        \begin{subfigure}[t]{\linewidth}
            \centering
            \includegraphics[width=\linewidth]{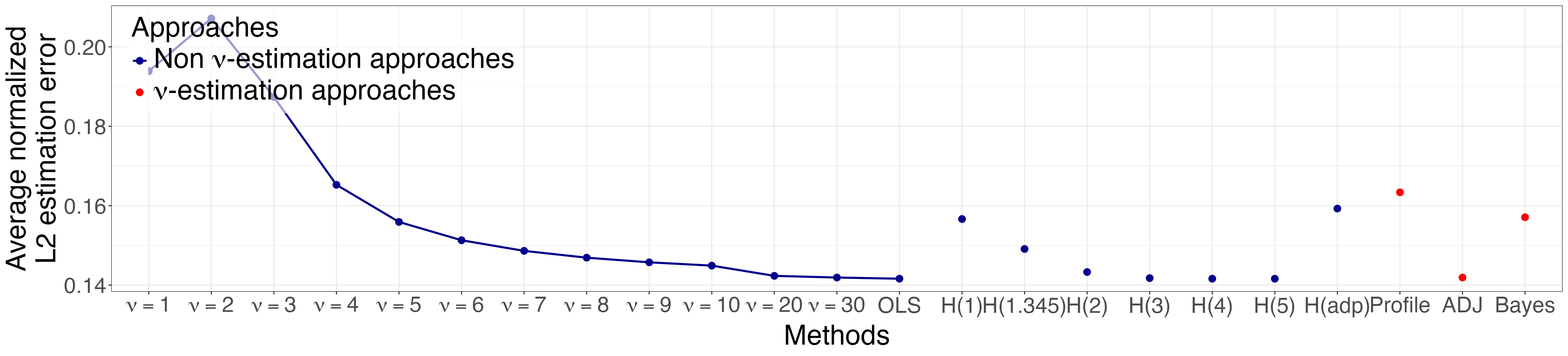}
            \caption{$n=100, p=50$}
            \label{fig:rmse_nerror_c}
        \end{subfigure}

    \end{minipage}
    }
    \caption{\com{Average $L_2$ error} of $\hat{\beta}$ with normally distributed errors. Fixed values of $\nu$ and other OLS and Huber estimates are shown in blue. The Huber constant $k$ was either fixed ($H(k)$) or estimated ($H({\rm adp})$). Estimates of $\nu$, either by maximum likelihood, $\hat\nu$, adjusted profile likelihood, $\hat\nu_{\rm adj}$, or MAP with Jeffreys prior, $\hat\nu_{\rm Bayes}$, are shown in red. }
    \label{fig:rmse_nerror_n20_p4_nu2}
\end{figure}

With standard normal errors, least squares estimates of $\beta$ have the smallest error, as expected (Figures~\ref{fig:rmse_nerror_a} and~\ref{fig:rmse_nerror_b}).  Estimation of $\nu$ seems to outperform adaptive Huber regression, with the $\hat{\nu}_{\rm adj}$ approach coming closest to OLS among the $\nu$-estimation approaches. Although \citet{Lange1989} recommended fixing $\nu$ for small sample sizes $n$, the adjusted profile likelihood approach yields performance comparable to the best fixed-$\nu$ choices across error settings, supporting its use for estimating $\nu$.

\section{Robustness to outliers} \label{robust}

\com{We considered linear models with $n=100$, $p=3$, and
$\beta=(1,5,10)$, and with $n=300$, $p=80$, and $\beta=0_p$.
The covariates were generated independently from
$\mathrm{Uniform}(0,1)$, while the uncontaminated errors
were generated from the standard normal distribution.}
For contamination rates $\eta=0.1,0.2,0.3$, we considered contamination of the error distribution by (i) $N(0, 9)$, (ii) $\chi^2_4 - 4$, or (iii) $t_2$.
We also considered (iv) two-point contamination, \com{in which the first $\eta n/2$ responses were set to $5$ and the next $\eta n/2$ responses were set to $-5$, while the remaining responses were generated from the uncontaminated model.}

\begin{figure}[h!]
    \centering

    \resizebox{0.9\textwidth}{!}{%
    \begin{minipage}{\textwidth}

        \begin{subfigure}[t]{\linewidth}
            \centering
            \includegraphics[width=\linewidth]{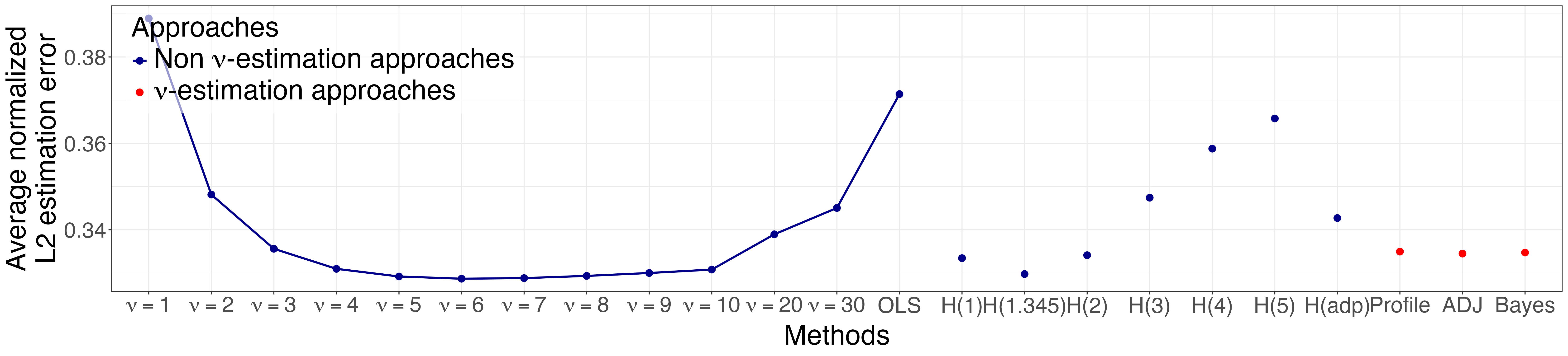}
            \caption{$10\%$ contamination}
        \end{subfigure}

        \vspace{1em}

        \begin{subfigure}[t]{\linewidth}
            \centering
            \includegraphics[width=\linewidth]{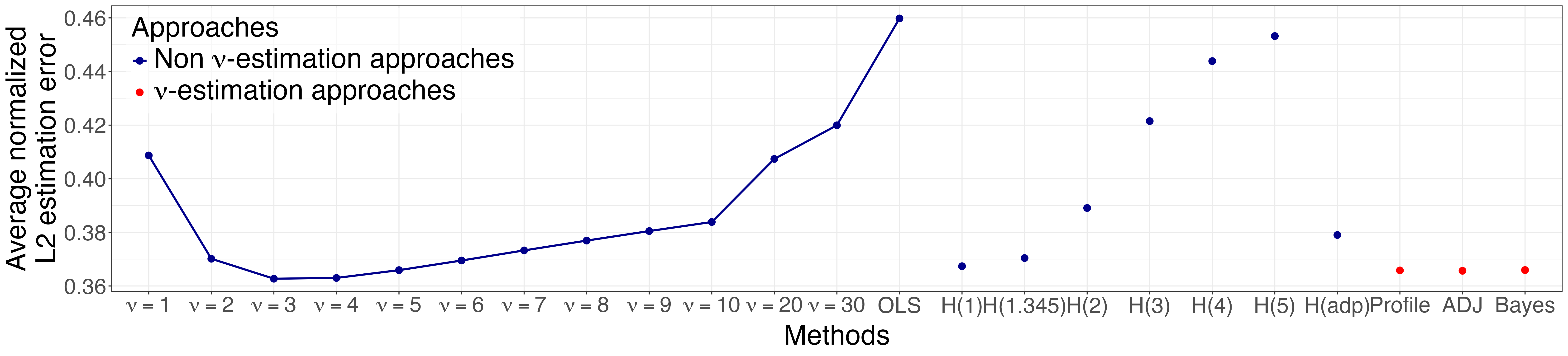}
            \caption{$20\%$ contamination}
        \end{subfigure}

        \vspace{1em}

        \begin{subfigure}[t]{\linewidth}
            \centering
            \includegraphics[width=\linewidth]{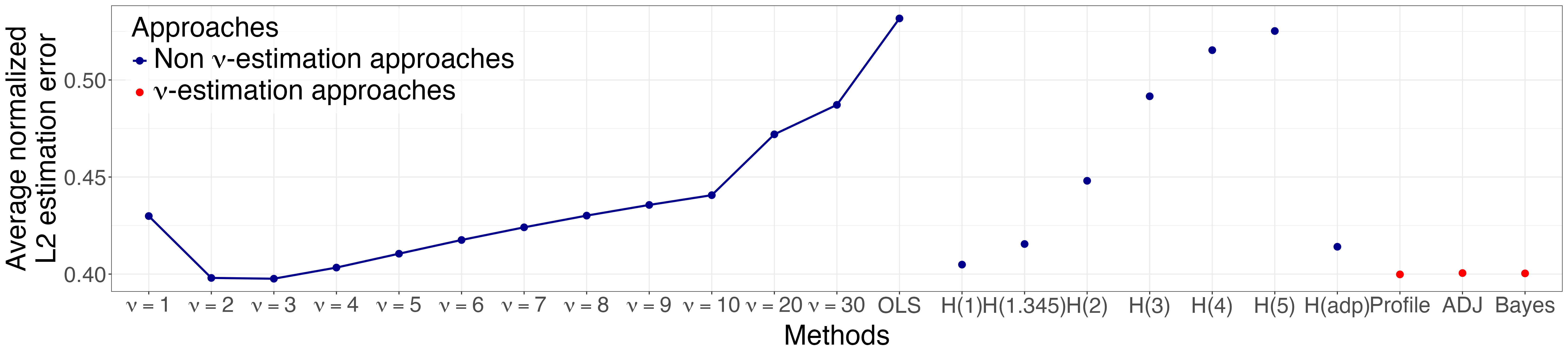}
            \caption{$30\%$ contamination}
        \end{subfigure}

    \end{minipage}
    }
      \caption{\com{Average $L_2$ error} of $\hat{\beta}$ with 10\%, 20\%, and 30\% $N(0,9)$ contaminated errors when $p = 3$. The plot compares non-$\nu$-estimation approaches (in blue) across various fixed values of $\nu$ and other traditional estimators (OLS and Huber) with the $\nu$-estimation approaches considered (in red). Here, $\nu = c$ denotes the approach in which $\nu$ is fixed at the constant value $c$; OLS denotes the ordinary least squares estimator; $H(c)$ denotes Huber regression with tuning parameter $k = c$; and $H(\rm adp)$ denotes adaptive Huber regression.}
    \label{fig:N(0,9)_10_20_30}
\end{figure}

\com{For $n=100$ and $p=3$, the maximum likelihood, adjusted
profile likelihood, and Bayesian MAP approaches to estimating $\nu$ give similar average $L_2$ errors for estimating $\beta$ under $N(0,9)$ contamination (Figure~\ref{fig:N(0,9)_10_20_30}). 
Their errors are close to the smallest errors among the fixed-$\nu$ values examined and are lower than those of OLS and adaptive Huber regression. 
The best-performing fixed value shifts towards smaller $\nu$ as the contamination rate increases.}

Errors from estimating $\nu$ are also close to the best fixed-$\nu$ results under $\chi^2_4-4$ and $t_2$ contamination (Figure~\ref{fig:chi4_10_20_30} and Figure~S9 in the supplementary material). 
Under two-point contamination, the three $\nu$-estimation approaches remain competitive in terms of the average $L_2$ error for $\beta$, whereas large fixed values of $\nu$ can lead to much larger errors (Figure~S8 in the supplementary material).
\com{Huber regression is sensitive to the tuning constant $k$, and adaptive selection does not consistently give smaller errors than the fixed-$k$ choices examined.}

\begin{figure}[h!]
    \centering
    \resizebox{0.9\textwidth}{!}{%
    \begin{minipage}{\textwidth}

        \begin{subfigure}[t]{\linewidth}
            \centering
            \includegraphics[width=\linewidth]{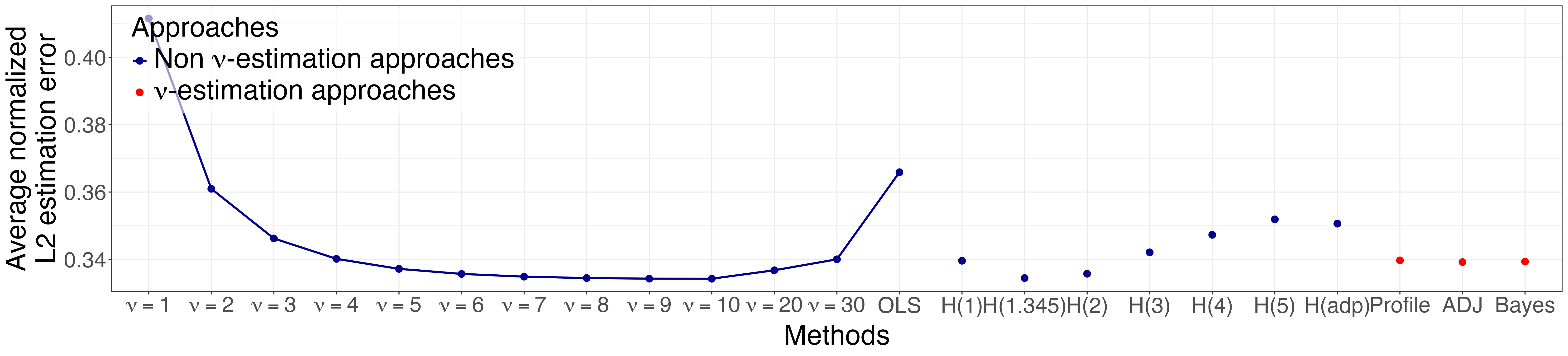}
            \caption{$10\%$ contamination}
        \end{subfigure}

        \vspace{1em}

        \begin{subfigure}[t]{\linewidth}
            \centering
            \includegraphics[width=\linewidth]{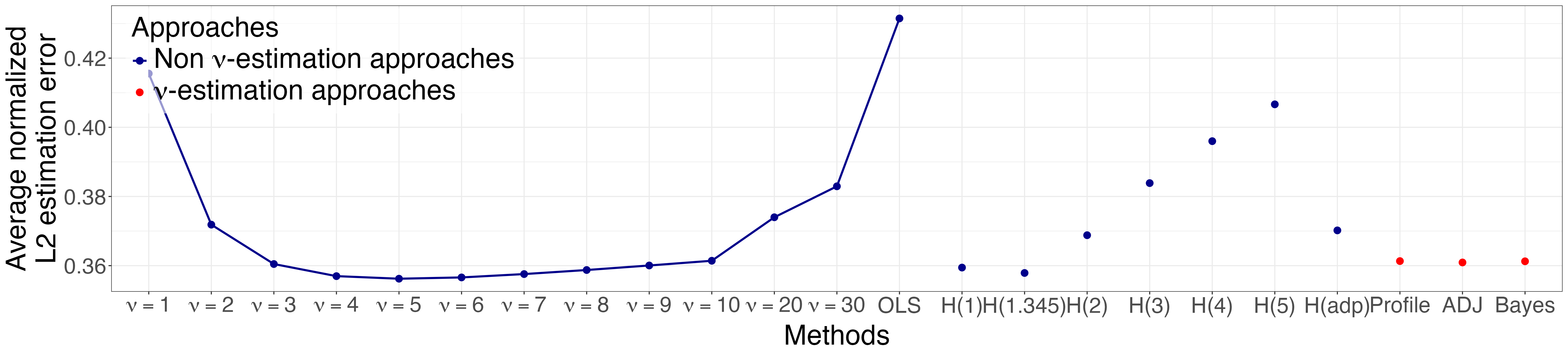}
            \caption{$20\%$ contamination}
        \end{subfigure}

        \vspace{1em}

        \begin{subfigure}[t]{\linewidth}
            \centering
            \includegraphics[width=\linewidth]{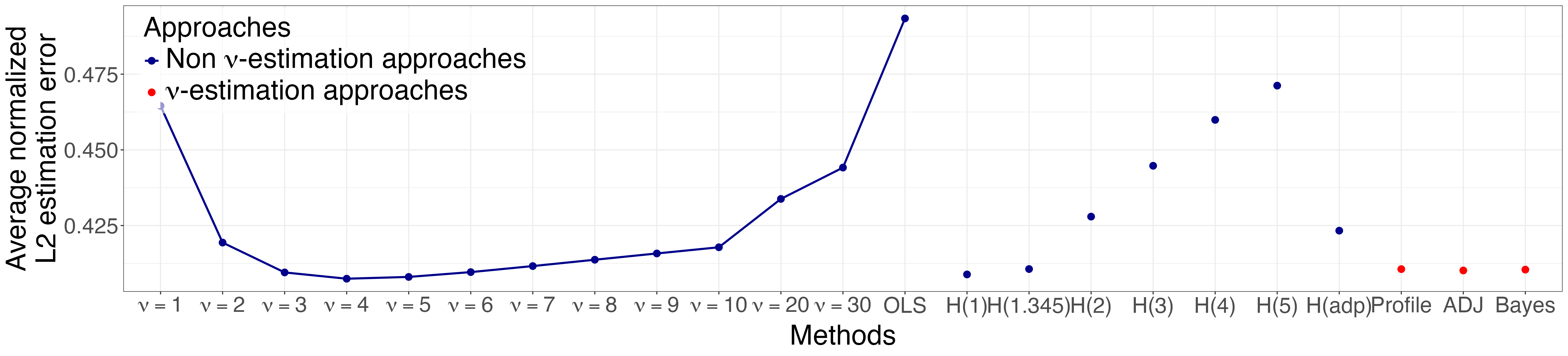}
            \caption{$30\%$ contamination}
        \end{subfigure}

    \end{minipage}
    }
      \caption{\com{Average $L_2$ error} of $\hat{\beta}$ with 10\%, 20\%, and 30\% $\chi^2_4 - 4$ contaminated errors when $p = 3$. The plot compares non-$\nu$-estimation approaches (in blue) across various fixed values of $\nu$ and other traditional estimators (OLS and Huber) with the $\nu$-estimation approaches considered (in red). Here, $\nu = c$ denotes the approach in which $\nu$ is fixed at the constant value $c$; OLS denotes the ordinary least squares estimator; $H(c)$ denotes Huber regression with tuning parameter $k = c$; and $H(\rm adp)$ denotes adaptive Huber regression.}
       \label{fig:chi4_10_20_30}
\end{figure}

\com{For $n=300$ and $p=80$, estimating $\nu$ by adjusted profile likelihood gives smaller average $L_2$ errors for estimating $\beta$ than using the maximum likelihood or Bayesian MAP estimate of $\nu$ under $N(0,9)$ contamination (Figure~\ref{fig:N(0,9)_10_20_30_high_dim}). }
Across all three contamination rates, the adjusted profile likelihood approach gives errors close to those of the best-performing fixed-$\nu$ and fixed-$k$ Huber fits examined, and lower than those of OLS and adaptive Huber regression. 
Similar relative performance among the three $\nu$-estimation approaches is observed under $\chi^2_4-4$ and $t_2$ contamination (Figures~S10 and~S11 in the supplementary material).

The results differ under two-point contamination (Figure~\ref{fig:2pt_10_20_30_high_dim}).
\com{At 10\% contamination, the adjusted profile likelihood
approach gives a lower average $L_2$ error than OLS and
adaptive Huber regression, although some small fixed
values of $\nu$ give smaller errors.
At 20\% contamination, the adjusted profile likelihood approach has larger error than adaptive Huber regression and the maximum likelihood and Bayesian MAP approaches, while all three $\nu$-estimation approaches still give smaller errors than OLS.
At 30\% contamination, the three approaches give errors similar to OLS, whereas adaptive Huber regression performs better.}
These results show that the benefit of estimating $\nu$ for regression coefficient estimation depends on the contamination mechanism and level.

\begin{figure}[h!]
    \centering
    \resizebox{0.9\textwidth}{!}{%
    \begin{minipage}{\textwidth}

        \begin{subfigure}[t]{\linewidth}
            \centering
            \includegraphics[width=\linewidth]{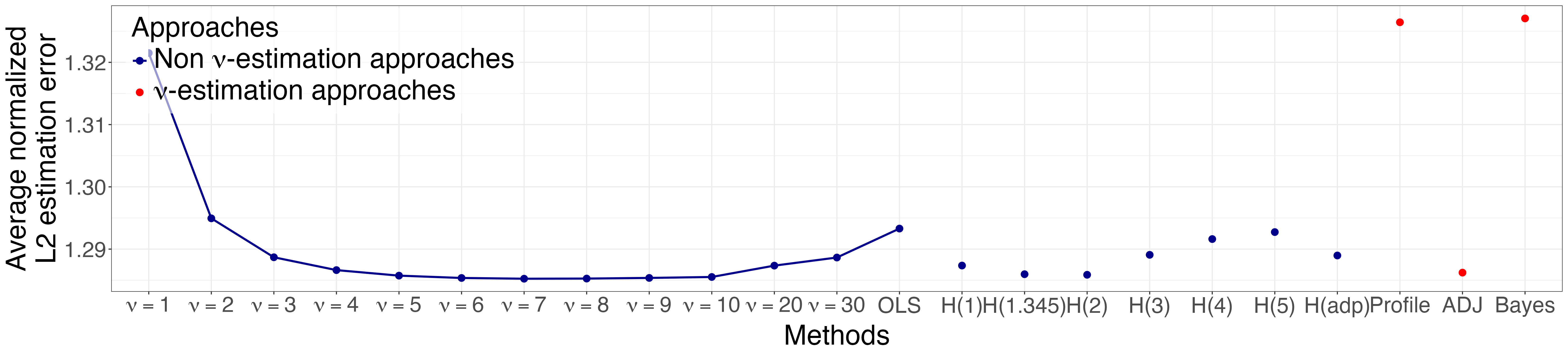}
            \caption{$10\%$ contamination}
        \end{subfigure}

        \vspace{1em}

        \begin{subfigure}[t]{\linewidth}
            \centering
            \includegraphics[width=\linewidth]{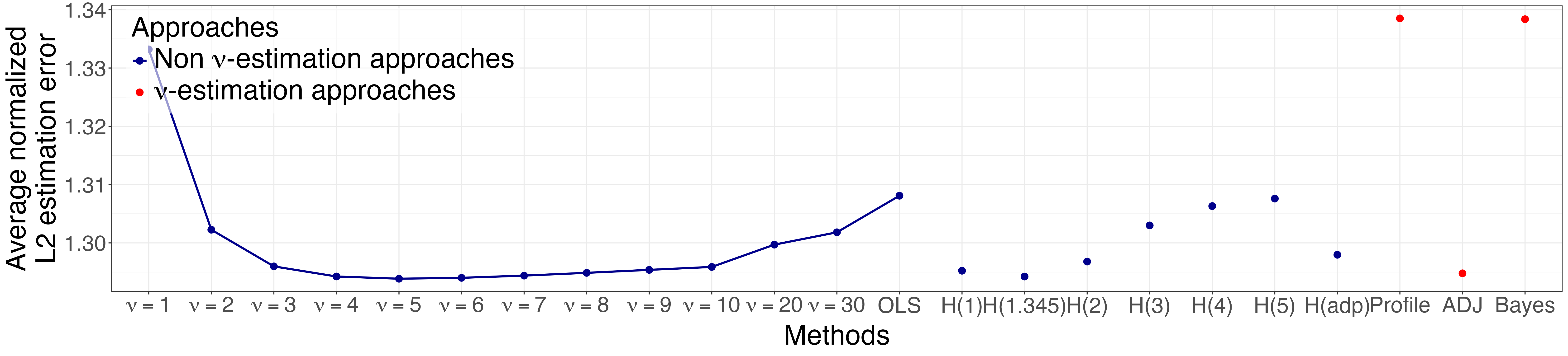}
            \caption{$20\%$ contamination}
        \end{subfigure}

        \vspace{1em}

        \begin{subfigure}[t]{\linewidth}
            \centering
            \includegraphics[width=\linewidth]{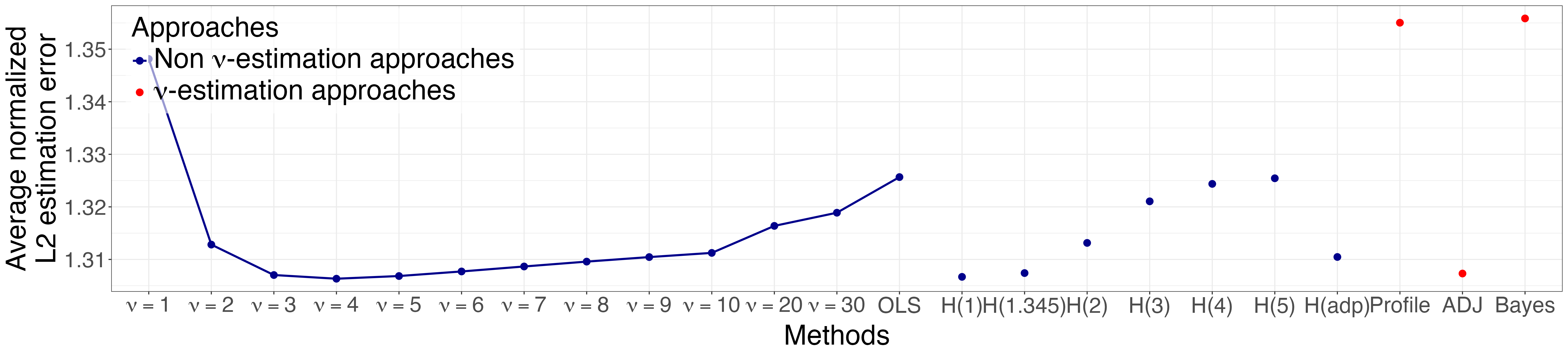}
            \caption{$30\%$ contamination}
        \end{subfigure}

    \end{minipage}
    }
      \caption{\com{Average $L_2$ error} of $\hat{\beta}$ with 10\%, 20\%, and 30\% $N(0,9)$ contaminated errors when $p = 80$. The plot compares non-$\nu$-estimation approaches (in blue) across various fixed values of $\nu$ and other traditional estimators (OLS and Huber) with the $\nu$-estimation approaches considered (in red). Here, $\nu = c$ denotes the approach in which $\nu$ is fixed at the constant value $c$; OLS denotes the ordinary least squares estimator; $H(c)$ denotes Huber regression with tuning parameter $k = c$; and $H(\rm adp)$ denotes adaptive Huber regression.}
      \label{fig:N(0,9)_10_20_30_high_dim}
\end{figure}

\begin{figure}[h!]
    \centering

    \resizebox{0.9\textwidth}{!}{%
    \begin{minipage}{\textwidth}

        \begin{subfigure}[t]{\linewidth}
            \centering
            \includegraphics[width=\linewidth]{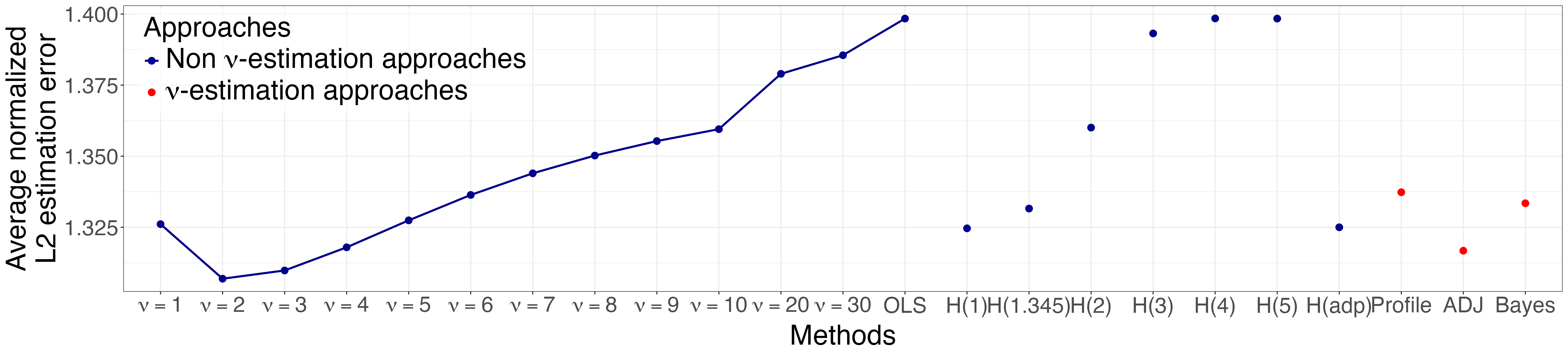}
            \caption{$10\%$ contamination}
        \end{subfigure}

        \vspace{1em}

        \begin{subfigure}[t]{\linewidth}
            \centering
            \includegraphics[width=\linewidth]{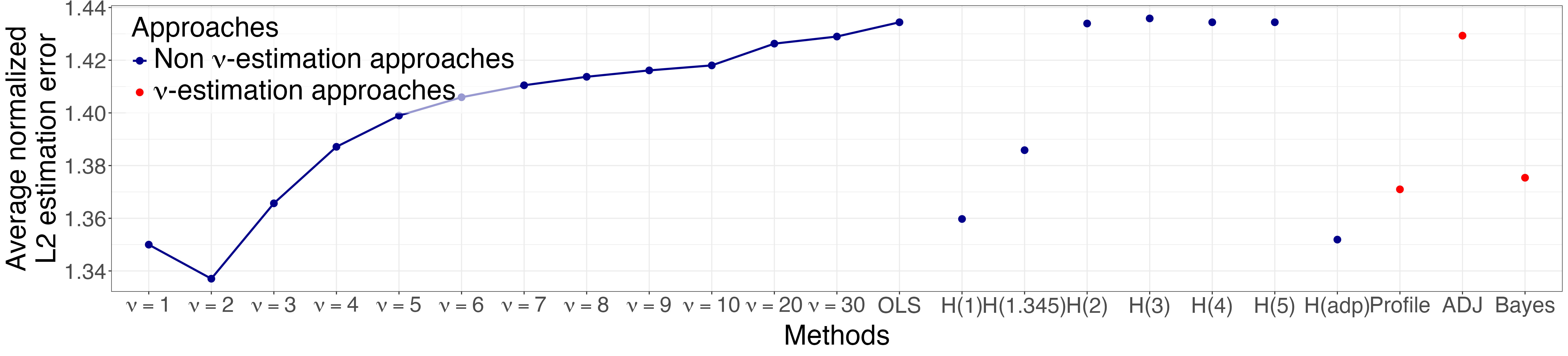}
            \caption{$20\%$ contamination}
        \end{subfigure}

        \vspace{1em}

        \begin{subfigure}[t]{\linewidth}
            \centering
            \includegraphics[width=\linewidth]{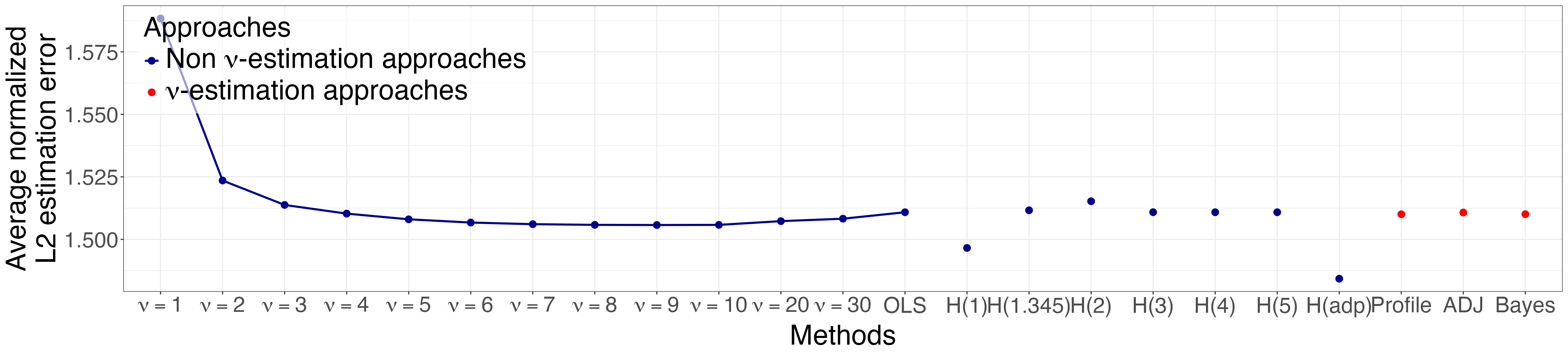}
            \caption{$30\%$ contamination}
        \end{subfigure}

    \end{minipage}
    }
      \caption{\com{Average $L_2$ error} of $\hat{\beta}$ with 10\%, 20\%, and 30\% two-point contaminated errors when $p = 80$. The plot compares non-$\nu$-estimation approaches (in blue) across various fixed values of $\nu$ and other traditional estimators (OLS and Huber) with the $\nu$-estimation approaches considered (in red). Here, $\nu = c$ denotes the approach in which $\nu$ is fixed at the constant value $c$; OLS denotes the ordinary least squares estimator; $H(c)$ denotes Huber regression with tuning parameter $k = c$; and $H(\rm adp)$ denotes adaptive Huber regression.}
          \label{fig:2pt_10_20_30_high_dim}
\end{figure}

\section{Conclusions}

In the context of linear regression with Student's $ t$-distributed errors, we addressed the commonly asked question of how to choose the degrees of freedom parameter $\nu$. We compared \com{likelihood-based and Bayesian MAP} estimators of $\nu$ and studied how these choices propagate to the estimation of the regression coefficient $\beta$ in uncontaminated and contaminated data. \com{Our results suggest that estimating $\nu$ by adjusted profile likelihood can be a useful practical option when the ratio $p/n$ is not too large. 
In many of the settings we considered, it gives competitive estimation accuracy for $\beta$ under both heavy-tailed and light-tailed errors, and under different forms of contamination. Yet, the method does not necessarily dominate fixed-$\nu$ fits and Huber regression. In some settings with large $p/n$ ratios or substantial model misspecification, fixed choices or Huber regression may perform better.} Since the proposed methods are sensitive to the $p/n$ ratio, we will leave a more precise characterization of this threshold for future work.

\bibliographystyle{apalike}
\bibliography{references2026}

\appendix
\section*{Appendix}
\section{Derivation of observed Fisher information for $\beta$ and $\sigma$}

\com{
This section derives the observed Fisher information matrix for $\lambda=(\beta,\sigma)$, treated as nuisance parameters when $\nu$ is the parameter of interest. This information matrix is used in the adjusted profile likelihood in the main text.
}
The log-likelihood function is
\begin{align*}
\ell(\beta, \sigma, \nu; y, x) &= c(\nu) - \frac{\nu+1}{2} \sum_{i = 1}^n \log\{\nu + (\frac{y_i - x_i^T\beta}{\sigma})^2 \} - n\log(\sigma).
\end{align*}
The first derivatives with respect to $\beta$ and $\sigma$ are given by
\begin{align*}
\frac{\partial \ell}{\partial \beta} 
&= (\nu + 1) \sum_{i = 1}^n x_i^T\frac{y_i - x_i^T\beta}{\nu\sigma^2 + (y_i - x_i^T\beta)^2}, \quad \quad \quad
\frac{\partial \ell}{\partial \sigma}  = -\frac{n}{\sigma} + \frac{\nu + 1}{\sigma} \sum_{i = 1}^n \frac{(y_i - x_i^T\beta)^2}{\nu\sigma^2 + (y_i - x_i^T\beta)^2}.
\end{align*}

 Define $r_i := y_i - x_i^T\beta$, then $z_i := r_i/\sigma$, and also define
 \begin{align*}
     g_i(\sigma):= \frac{r_i^2}{\nu\sigma^2 + r_i^2}, 
     \quad \quad \quad 
     g_i'(\sigma) 
 = \frac{-2\sigma\nu r_i^2}{(\sigma^2\nu + r_i^2)^2}.
 \end{align*}

The observed Fisher information matrix entries for $\theta = (\beta, \sigma)$ are 
\begin{align*}
    \frac{\partial ^2 \ell}{\partial \sigma^2} &= \frac{n}{\sigma^2} - (\nu + 1)\sum_{i = 1}^n[\frac{z_i^2}{\sigma^2\nu + r_i^2} + \frac{2\nu r_i^2}{(\sigma^2\nu + r_i^2)^2}], 
    \quad \quad 
    \frac{\partial ^2 \ell}{\partial \beta^2} =  \frac{\nu+1}{\sigma^2}\sum_{i = 1}^n \frac{-x_ix_i^T}{\nu + z_i^2}+ \frac{2z_i^2 x_i x_i^T}{(\nu + z_i^2)^2}, \\
    \frac{\partial ^2 \ell}{\partial \beta \partial \sigma} &=  -\frac{(\nu + 1)2\nu}{\sigma^2}\sum_{i = 1}^n \frac{x_iz_i}{(\nu + z_i^2)^2}. 
\end{align*}

\section{Numerical difficulties in $\nu$ optimization}

The numerical optimization of $\nu$ can be challenging in some settings, particularly because of boundary behaviour and relatively flat or numerically unstable objective functions. This section provides additional details on these numerical difficulties, including boundary handling, the choice of sample sizes used in the simulation study, initialization, and numerical diagnostics.

\subsection{Boundary handling in the $\omega$-parameterization}

Using the $\omega = 1/ \nu$ parametrization, for any $\omega<0$, we set the objective functions to return $-\infty$ to enforce the constraint $\omega \geq 0$. When $\omega=0$, we treated the error distribution as normal when evaluating the likelihood component of the objective functions. Finally, we inverted $\hat{\omega}$ to obtain $\hat{\nu}$. 

\subsection{Choice of sample sizes and initialization}
\label{supp_sec:numerical_issues}

\com{
\citet{Fonseca2008} discussed difficulties in obtaining finite
estimates of $\nu$ associated with profile likelihood behaviour near the
normal-model limit $\nu\to\infty$. Their diagnostic condition is
$\sum_{i=1}^n(\hat{z}_i^2-1)^2<2n$, where $\hat{z}_i=(y_i-x_i^T\hat{\beta})/\hat{\sigma}$ and $\hat{\beta}$ and $\hat{\sigma}$ are the maximum likelihood estimates under the normal error model.
Since $\nu\to\infty$ corresponds to $\omega\to0$, this condition concerns the behaviour of the profile likelihood near the boundary $\omega=0$ under our parametrization.
In preliminary simulations, we used this condition to guide sample size selection and focused on settings in which it was not satisfied. 
We also visually inspected the four objective functions
for a number of simulated datasets to assess potential numerical difficulties, including relative flatness and instability.
These checks informed our choice of relatively favourable optimization settings, but did not guarantee successful convergence.
Figure~\ref{fig:flatness} illustrates a scenario in which the profile likelihood and
the two Bayesian-type objective functions are relatively flat over
parts of the plotted range, whereas the adjusted profile likelihood
exhibits a clear interior peak.
}

\com{
In additional simulations with $n=100$ and true $\nu=2$, optimization failures occurred in approximately $3.4\%$ and $66.8\%$ of repetitions for $p=40$ and $p=60$, respectively.
Among the first 500 repetitions with optimizer-reported successful convergence, the proportions with $\hat{\nu}_{\rm adj}>10$ were approximately $26.20\%$ and
$17.80\%$, respectively.
Figure~\ref{fig:n100_boxplots} shows the distributions of the estimates from the first 500 repetitions in which the optimizer reported successful convergence.
However, the relatively concentrated distributions for the other
three approaches do not necessarily indicate reliable optimization.
In Figure~\ref{fig:high-dim-issue-successful}, the optimizer reported successful convergence for all four approaches. The adjusted profile likelihood exhibits a clear interior peak, with the estimate returned by the optimizer close to the peak (see Figure \ref{fig:high-dim-issue-successful_b}).
In contrast, the objective functions for the other three approaches
appear to exhibit numerically unstable behaviour in Figures~\ref{fig:high-dim-issue-successful} and~\ref{fig:high-dim-issue-failed}, and their returned estimates do not appear to maximize the corresponding objective functions over the plotted ranges. The boxplots should therefore be interpreted alongside these numerical diagnostics.}

\com{
Examination of a number of unsuccessful adjusted profile likelihood runs with $p=60$ identified the same boundary-related failure mechanism.
The line search accepted a small positive value of $\omega$
close to zero, but the subsequent numerical finite-difference
gradient calculation evaluated the objective function being
optimized at a negative value of $\omega$.
This evaluation returned $-\infty$, producing a non-finite
gradient approximation and causing the optimizer to terminate
with an error, as described in Section~2.3 of the main text.
Figure~\ref{fig:high-dim-issue-failed-b} illustrates an unsuccessful case in which the adjusted profile log-likelihood takes its largest plotted values near $\omega=0$.}

\com{
These numerical difficulties motivated our choice of relatively large sample sizes in the main simulation study. The results should therefore be interpreted as comparisons under relatively favourable
optimization conditions.
}

\com{
We also examined sensitivity to initialization in the setting with
$n=300$, $p=60$, and true $\nu=2$.
In addition to the data-dependent initialization used in the main text, we considered fixed starting values $\omega_{\text{init}}=1/2$, $1/5$, and $1/10$, where $1/2$ is the true value of $\omega$.
We compared the resulting distributions of the estimates and their root mean square errors (RMSEs).
Figure~\ref{fig:omega_init_comparison} shows the distributions of the
estimates from repetitions with optimizer-reported successful
convergence under these initialization strategies.
}

\com{
The RMSE of the adjusted profile likelihood estimator was relatively
insensitive to the starting values examined.
For $\omega_{\text{init}}=1/5$, the RMSE of $\hat{\nu}_{\rm adj}$ was $0.5824$, close to that obtained using $\omega_{\text{init}}=1/2$.
The other three approaches gave RMSE values of approximately
$1.36$--$1.39$, compared with $1.25$--$1.29$ when initialized at
the true value of $\omega$.
For $\omega_{\text{init}}=1/10$, the RMSE of $\hat{\nu}_{\rm adj}$ was $0.5860$, while the other approaches gave RMSE values of approximately $1.44$--$1.49$.
For the data-dependent initialization used in the main text, the
corresponding RMSEs were 1.4456, 0.5838, 1.4617, and 1.4377 for
$\hat{\nu}$, $\hat{\nu}_{\rm adj}$, $\hat{\nu}_{\rm Bayes}$, and
$\hat{\nu}_{\rm Pseudo}$, respectively.
The profile likelihood and the two Bayesian-type approaches tended to underestimate $\nu$, with greater
underestimation as the fixed starting value of $\omega$ decreased.
The adjusted profile likelihood approach gave the smallest RMSE
under each of the initialization strategies examined, including
the data-dependent initialization used in the main text.
}

\begin{figure}[h!] \centering 
    \begin{subfigure}{0.48\textwidth} \centering \includegraphics[width=\linewidth]{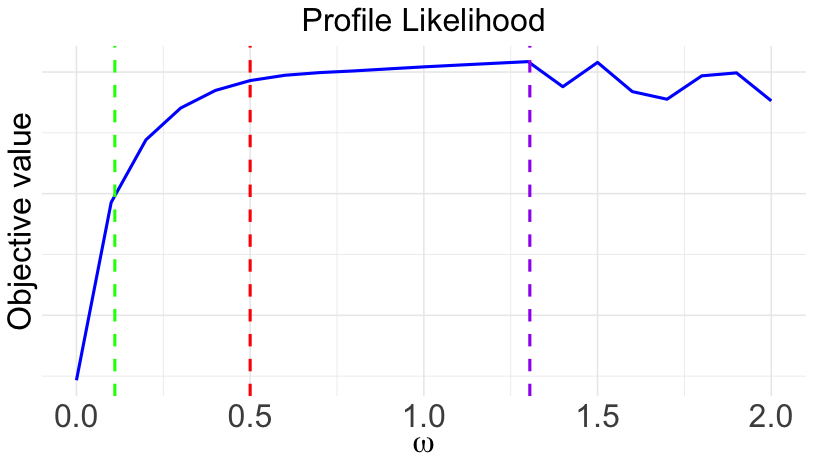} \caption{Objective function of Profile likelihood approach}
    \end{subfigure} 
    \hfill 
    \begin{subfigure}{0.48\textwidth} \centering \includegraphics[width=\linewidth]{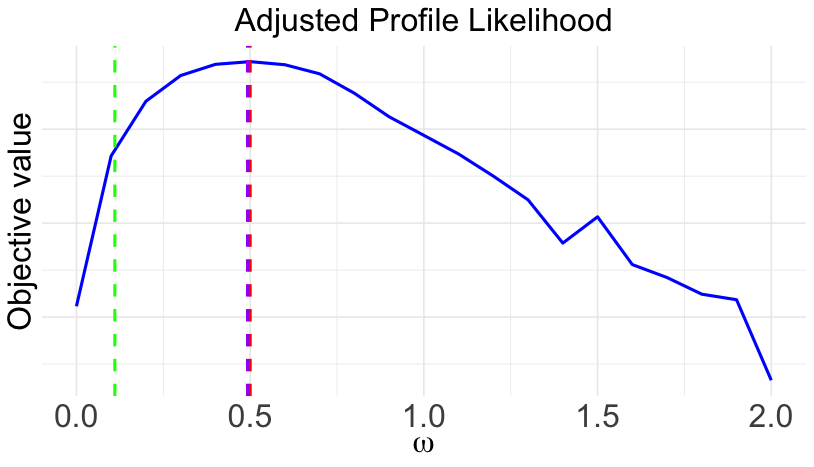} \caption{Objective function of Adjusted profile likelihood approach} 
    \end{subfigure} 
    \vspace{0.5em} 
    \begin{subfigure}{0.48\textwidth} \centering \includegraphics[width=\linewidth]{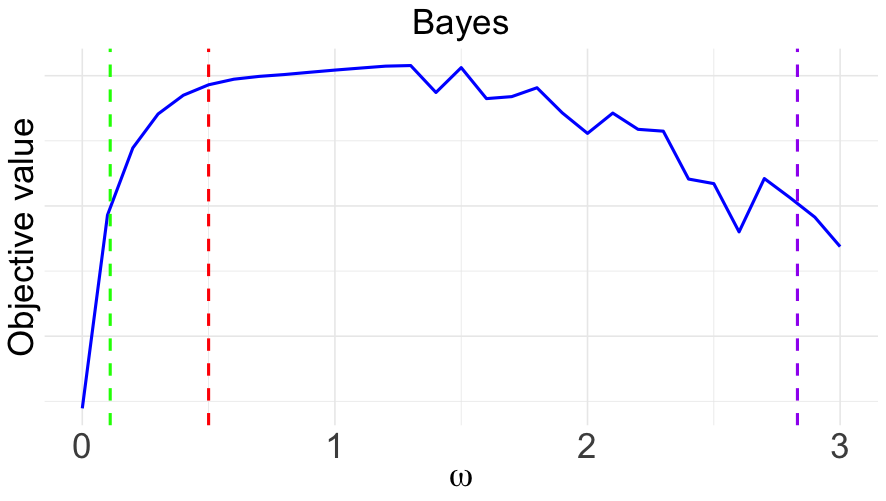} \caption{Objective function of Bayes approach} 
    \end{subfigure} 
    \hfill 
    \begin{subfigure}{0.48\textwidth} \centering \includegraphics[width=\linewidth]{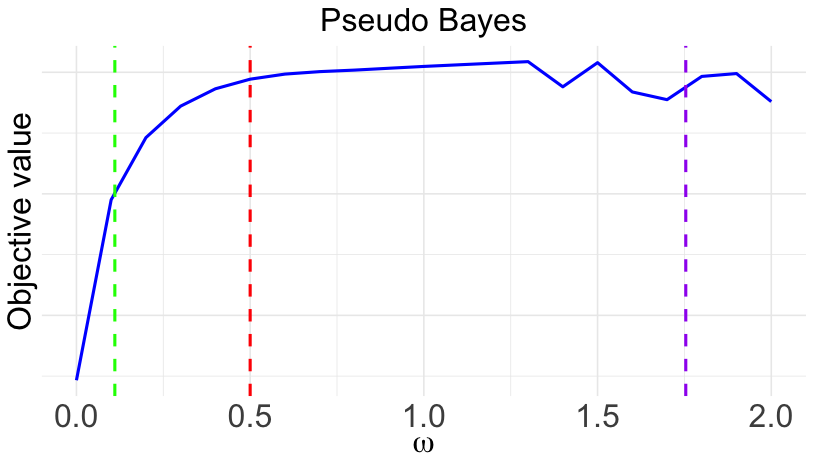} \caption{Objective function of Pseudo Bayes approach} 
    \end{subfigure} 
    \caption{\com{Objective functions for a simulated dataset with
    $n=300$, $p=60$, true $\nu=2$, and data-generation seed 98.
    The adjusted profile likelihood exhibits a clear interior peak, with the estimate returned by the optimizer close to the peak.
    In contrast, the objective functions for the profile likelihood and the two Bayesian MAP-based approaches are relatively flat over parts of the plotted range.
    The red lines indicate the true value $\omega=0.5$, the green lines indicate the data-dependent initial $\omega$ value, and the purple lines indicate the estimates returned by the optimizer for the corresponding approach.}}
\label{fig:flatness}
\end{figure}

\begin{figure}[h!]
    \centering
    \begin{subfigure}[t]{0.48\textwidth}
        \centering
        \includegraphics[width=\linewidth]{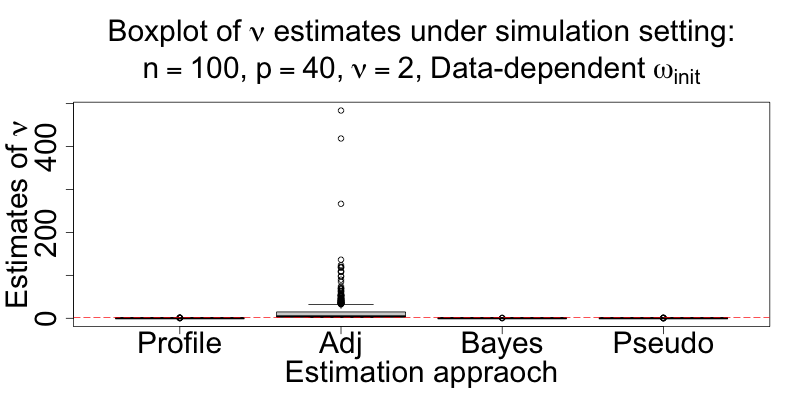}
        \caption{$p=40$}
    \end{subfigure}
    \hfill
    \begin{subfigure}[t]{0.48\textwidth}
        \centering
        \includegraphics[width=\linewidth]{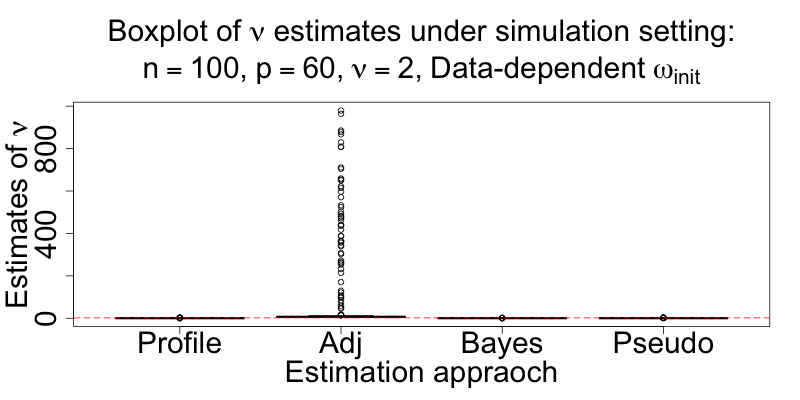}
        \caption{$p=60$}
    \end{subfigure}
    \caption{\com{Boxplots of $\nu$ estimates under different approaches for (a) $p=40$ and (b) $p=60$, with $n=100$, true $\nu=2$, and data-dependent initialization. The boxplots are based on the first 500 repetitions with optimizer-reported successful convergence for each setting.}}
    \label{fig:n100_boxplots}
\end{figure}

\begin{figure}[h!] \centering 
    \begin{subfigure}{0.48\textwidth} \centering \includegraphics[width=\linewidth]{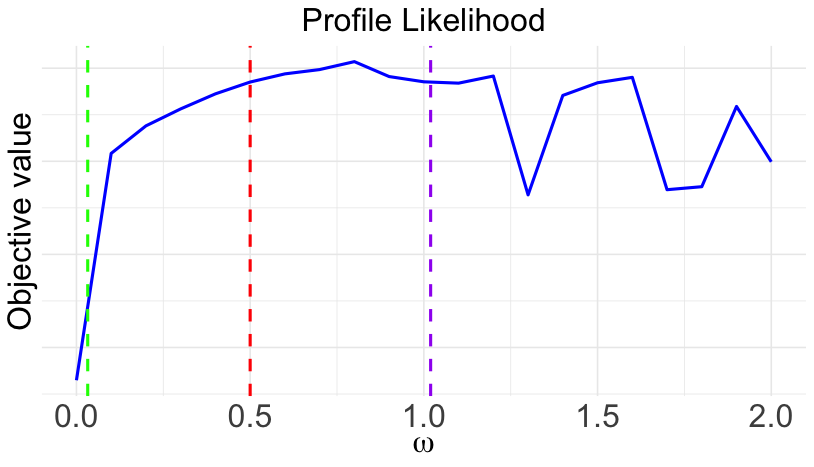} \caption{Objective function of Profile likelihood approach}
    \end{subfigure} 
    \hfill 
    \begin{subfigure}{0.48\textwidth} \centering \includegraphics[width=\linewidth]{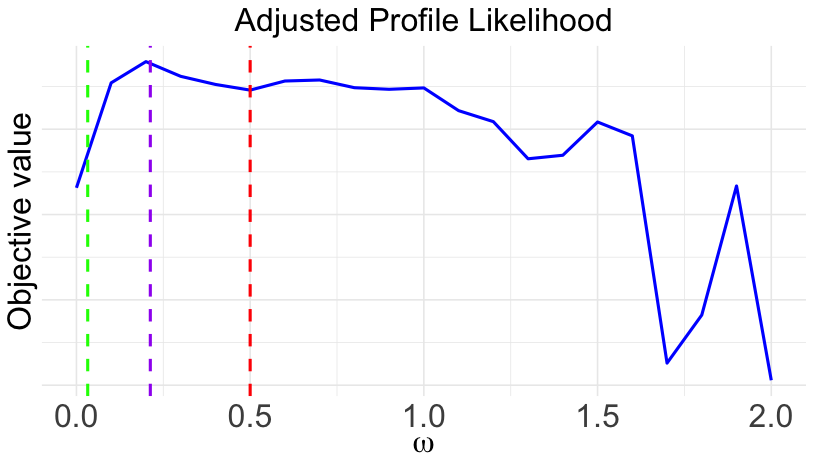} \caption{Objective function of Adjusted profile likelihood approach} 
    \label{fig:high-dim-issue-successful_b}
    \end{subfigure} 
    \vspace{0.5em} 
    \begin{subfigure}{0.48\textwidth} \centering \includegraphics[width=\linewidth]{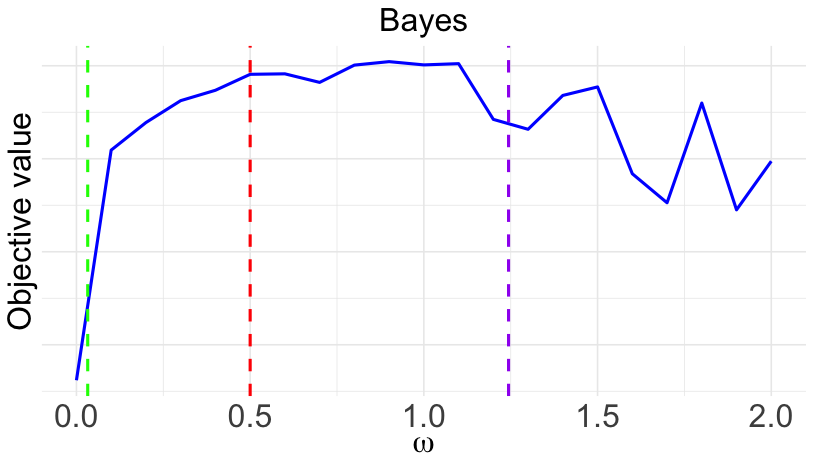} \caption{Objective function of Bayes approach} 
    \end{subfigure} 
    \hfill 
    \begin{subfigure}{0.48\textwidth} \centering \includegraphics[width=\linewidth]{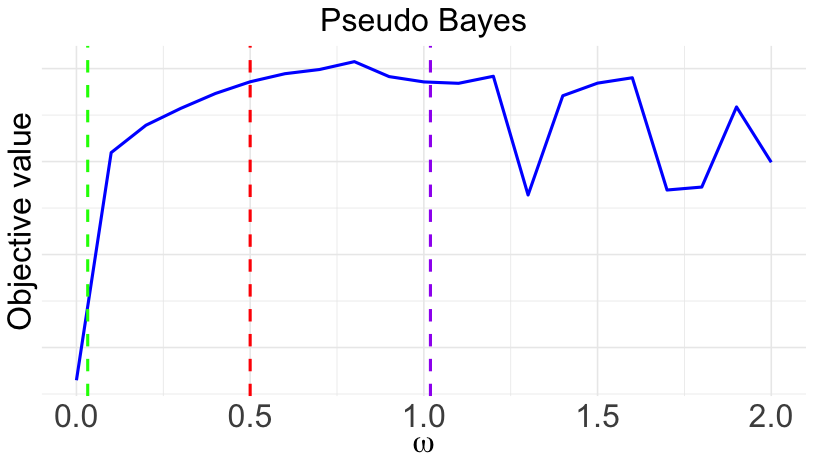} \caption{Objective function of Pseudo Bayes approach} 
    \end{subfigure} 
    \caption{\com{Objective functions for a simulated dataset with $n=100$, $p=60$, true $\nu=2$, and data-generation seed 560. 
    The optimizer reported successful convergence for all four approaches. The adjusted profile likelihood estimate appears to maximize its objective function over the plotted range. 
    However, the other three objective functions appear to exhibit numerically unstable behaviour, and the estimates returned for these approaches do not appear to maximize the corresponding objective functions.
    The red lines indicate the true value $\omega =0.5$, the green lines indicate the data-dependent initial $\omega$ value, and the purple lines indicate the estimates returned by the optimizer for the corresponding approach.}} 
\label{fig:high-dim-issue-successful}
\end{figure}

\begin{figure}[h!] \centering 
    \begin{subfigure}{0.48\textwidth} \centering \includegraphics[width=\linewidth]{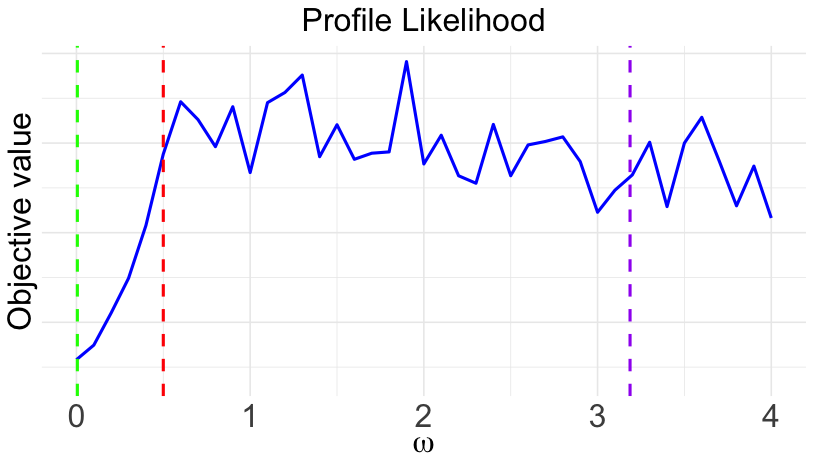} \caption{Objective function of Profile likelihood approach}
    \label{fig:high-dim-issue-failed-a}
    \end{subfigure} 
    \hfill 
    \begin{subfigure}{0.48\textwidth} \centering \includegraphics[width=\linewidth]{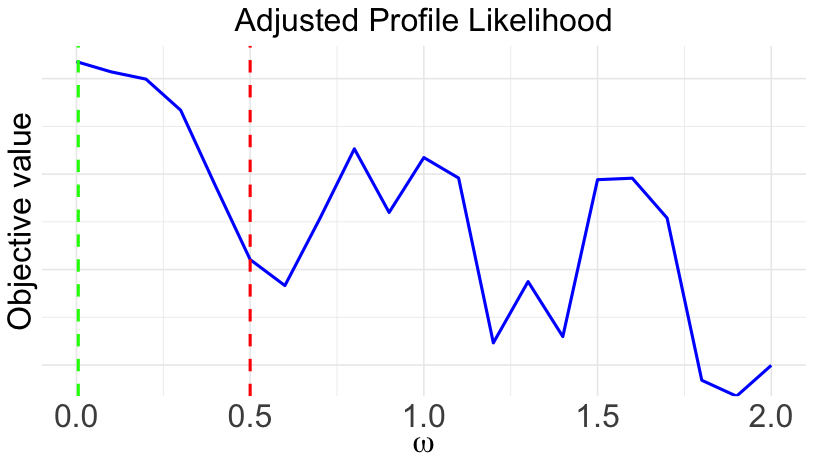} \caption{Objective function of Adjusted profile likelihood approach}
    \label{fig:high-dim-issue-failed-b}
    \end{subfigure} 
    \vspace{0.5em} 
    \begin{subfigure}{0.48\textwidth} \centering \includegraphics[width=\linewidth]{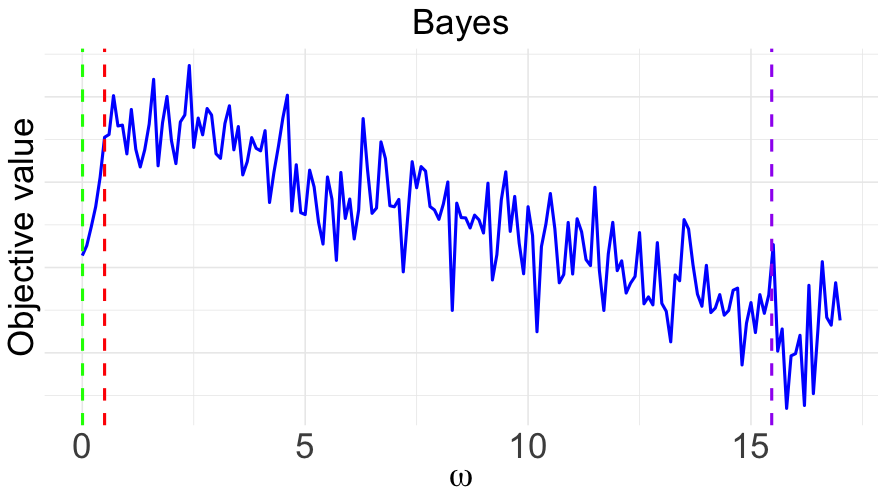} \caption{Objective function of Bayes approach} 
    \label{fig:high-dim-issue-failed-c}
    \end{subfigure} 
    \hfill 
    \begin{subfigure}{0.48\textwidth} \centering \includegraphics[width=\linewidth]{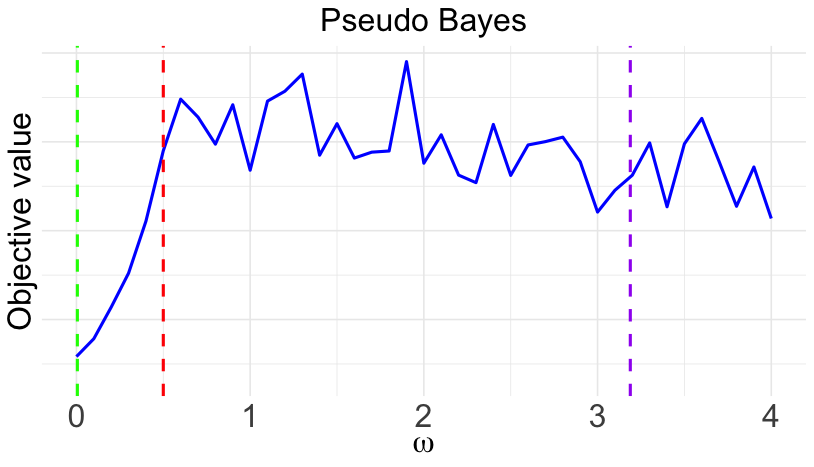} \caption{Objective function of Pseudo Bayes approach} 
    \label{fig:high-dim-issue-failed-d}
    \end{subfigure} 
\caption{\com{
Objective functions for a simulated dataset with $n=100$, $p=60$,
true $\nu=2$, and data-generation seed 133. The adjusted profile
likelihood optimization terminated with an error; panel (b) shows
that its objective function takes its largest plotted values near
$\omega=0$. The boundary-related failure mechanism is discussed
in Section~\ref{supp_sec:numerical_issues}.
The other three objective functions exhibit numerically unstable
behaviour, and the estimates returned for these approaches do not appear to maximize the corresponding objective functions. The red lines indicate the true value $\omega =0.5$, the green lines indicate the data-dependent initial $\omega$ value, and the purple lines indicate the estimates returned by the optimizer
where available. 
}} 
\label{fig:high-dim-issue-failed}
\end{figure}

\begin{figure}[h!]
    \centering
    \begin{subfigure}[t]{0.48\textwidth}
        \centering
        \includegraphics[width=\linewidth]{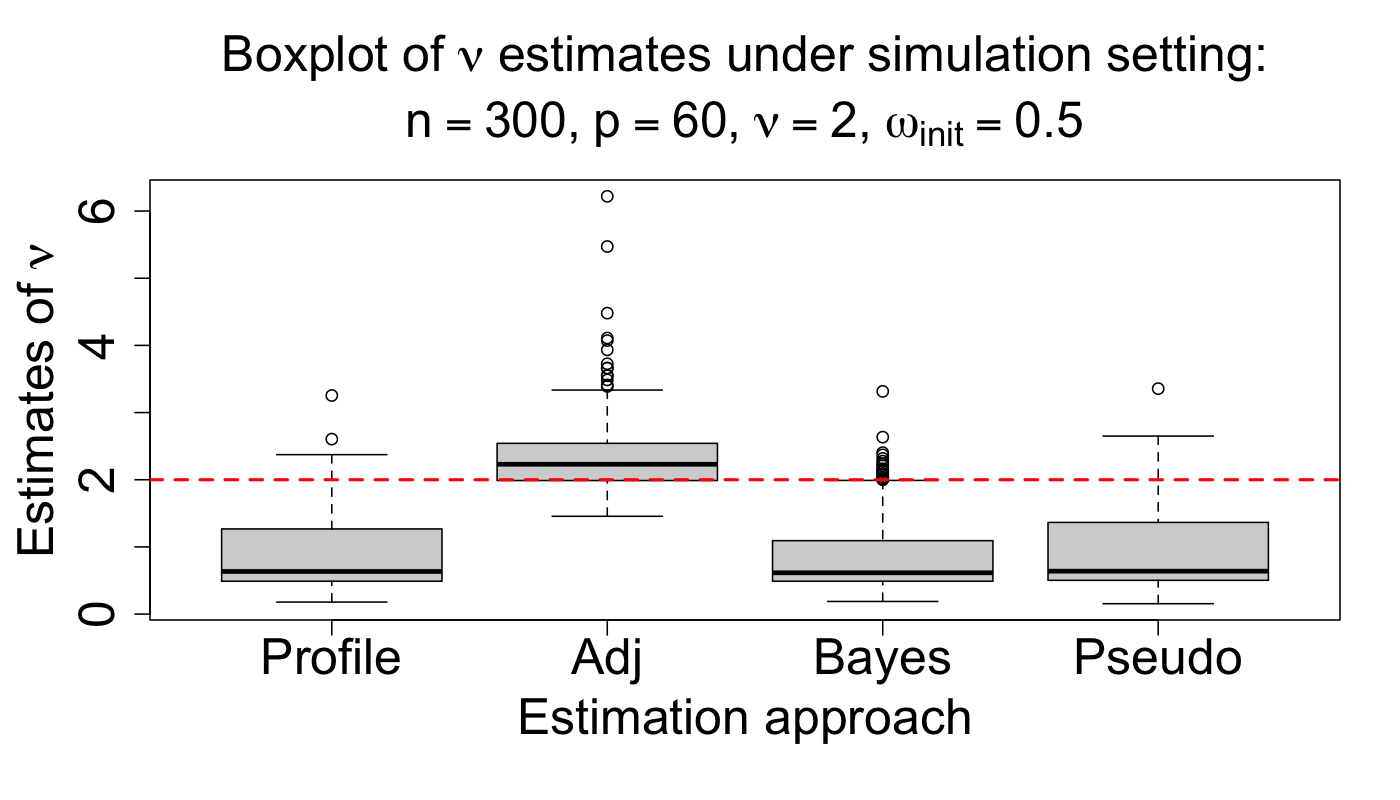}
        \caption{$\omega_{\text{init}} = 1/2$}
    \end{subfigure}
    \hfill 
    \begin{subfigure}[t]{0.48\textwidth}
        \centering
        \includegraphics[width=\linewidth]{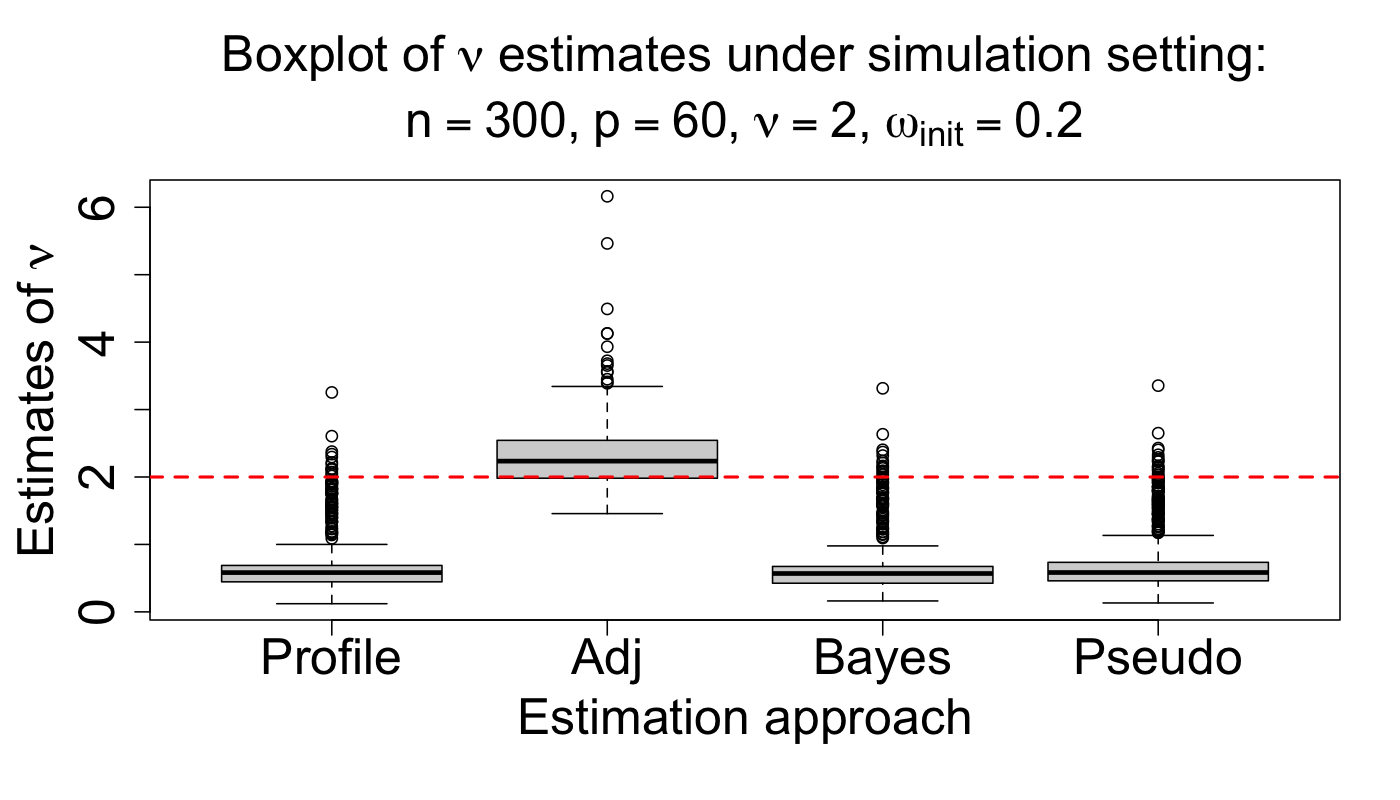}
        \caption{$\omega_{\text{init}} = 1/5$}
    \end{subfigure}
    \vspace{0.5em} 
    \begin{subfigure}[t]{0.48\textwidth}
        \centering
        \includegraphics[width=\linewidth]{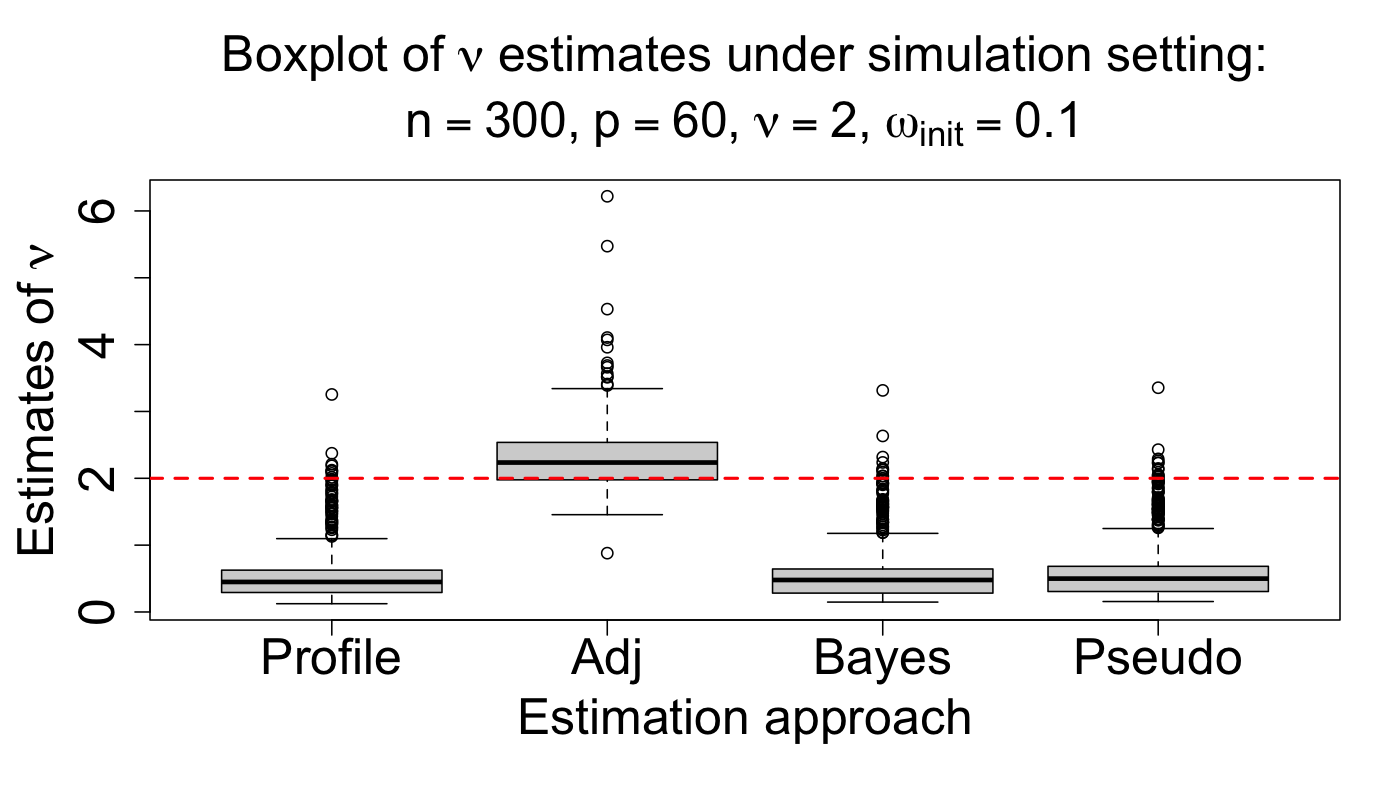}
        \caption{$\omega_{\text{init}} = 1/10$}
    \end{subfigure}
    \hfill
    \vspace{0.5em} 
    \begin{subfigure}[t]{0.48\textwidth}
        \centering
        \includegraphics[width=\linewidth]{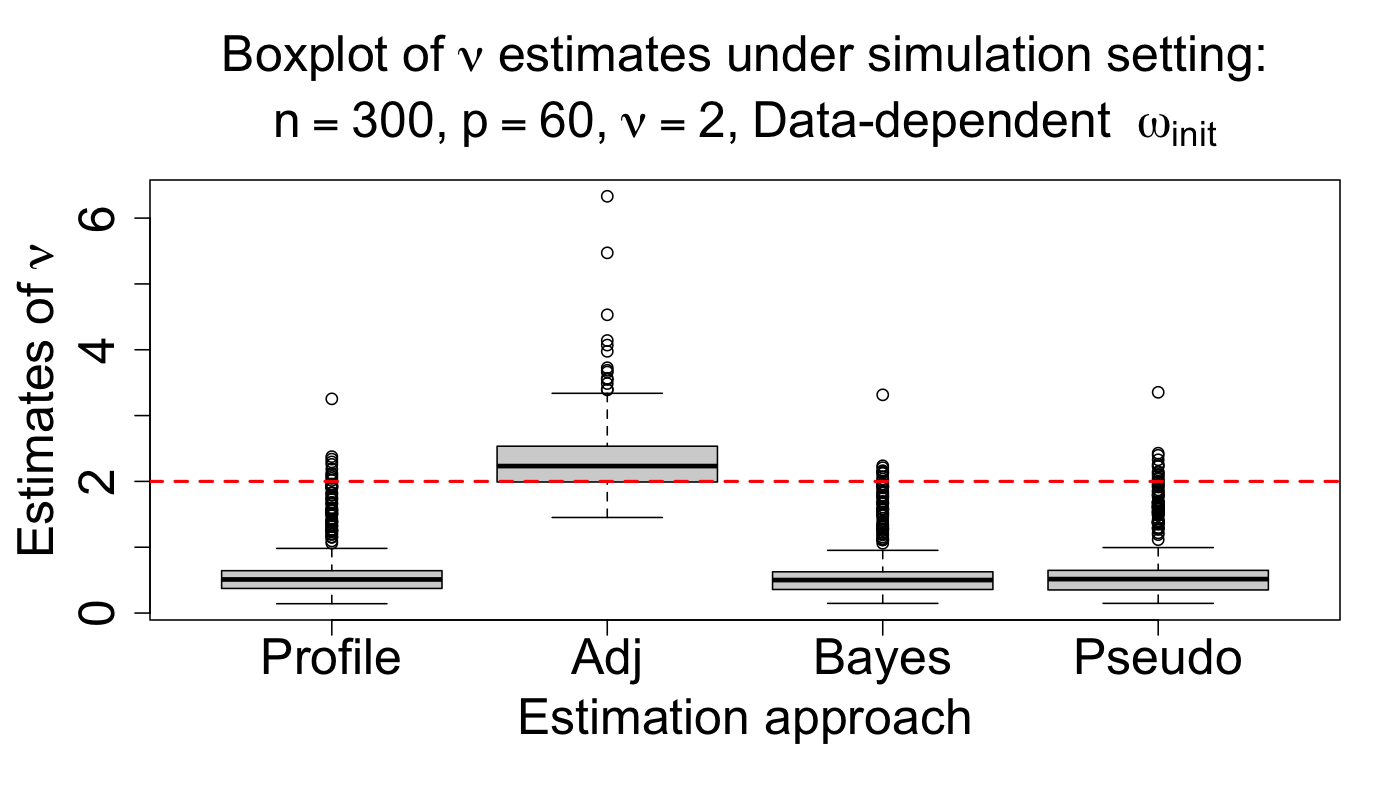}
        \caption{Data-dependent $\omega_{\text{init}}$}
    \end{subfigure}
    
    \caption{\com{Boxplots of $\nu$ estimates obtained using different initialization strategies for $n=300$, $p=60$, and true $\nu=2$:
    (a) $\omega_{\text{init}}=1/2$ (the true value of $\omega$),
    (b) $\omega_{\text{init}}=1/5$,
    (c) $\omega_{\text{init}}=1/10$, and
    (d) data-dependent initialization.
    Only simulations in which the optimizer reported successful convergence are included. The red line indicates the true $\nu$ value.}}
    \label{fig:omega_init_comparison}
\end{figure}

\clearpage

\section{Additional simulation results for $\nu$ estimation}

\com{
In this section, we present additional results for the estimation of $\nu$.
The RMSEs of the $\nu$ estimates for true
$\nu=2$ and $n=300$ are reported in Table~1 of the main text.
Tables~\ref{tab:rmse_n2500_nu5}--\ref{tab:rmse_n4500_nu10} report the corresponding RMSE results for true $\nu=5$ and $10$, with $n=2500$ and $4500$, respectively. 
Tables~\ref{tab:bias_n300_nu2}--\ref{tab:bias_n4500_nu10} and \ref{tab:se_n300_nu2}--\ref{tab:se_n4500_nu10} report the bias and standard error, respectively, for true $\nu=2,5,10$.
}

\com{
We also report the median numerical approximations to the first and second derivatives of the objective functions with respect to $\omega$ at the estimated optima in Tables~\ref{tab:curvature_n300_nu2}--\ref{tab:derivative_n4500_nu10}. 
For $n=300$ and true $\nu=2$, these numerical derivatives become particularly large at $p=60$ and $80$ for the profile likelihood and the two Bayesian-typeapproaches, while remaining substantially smaller for the adjusted profile likelihood. 
Figure~\ref{fig:derivative} illustrates the corresponding numerical instability for a representative simulated dataset with $p=60$.
}

\begin{table*}[h!]%
\centering %
\caption{\com{RMSE of estimates of $\nu$}, with true $\nu = 5$ and sample size $n = 2500$.\label{tab:rmse_n2500_nu5}}%
\begin{tabular*}{\textwidth}{@{\extracolsep\fill}rrrrr@{\extracolsep\fill}}
\hline
\textbf{$p$} & \textbf{$\hat{\nu}$} & \textbf{$\hat{\nu}_{\rm adj}$} & \textbf{$\hat{\nu}_{\rm Bayes}$} & \textbf{$\hat{\nu}_{\rm Pseudo}$} \\
\hline
1 & 0.5406 & 0.5412 & 0.5441 & 0.5450 \\
2 & 0.5404 & 0.5416 & 0.5438 & 0.5447 \\
5 & 0.5388 & 0.5418 & 0.5421 & 0.5429 \\
10 & 0.5391 & 0.5450 & 0.5420 & 0.5429 \\
20 & 0.5359 & 0.5464 & 0.5384 & 0.5391 \\
40 & 0.5372 & 0.5538 & 0.5387 & 0.5391 \\
60 & 0.5376 & 0.5571 & 0.5381 & 0.5382 \\
80 & 0.5432 & 0.5555 & 0.5425 & 0.5423 \\
\hline
\end{tabular*}
\end{table*}

\begin{table*}[h!]%
\centering %
\caption{\com{RMSE of estimates of $\nu$}, with true $\nu = 10$ and sample size $n = 4500$.\label{tab:rmse_n4500_nu10}}%
\begin{tabular*}{\textwidth}{@{\extracolsep\fill}rrrrr@{\extracolsep\fill}}
\hline
\textbf{$p$} & \textbf{$\hat{\nu}$} & \textbf{$\hat{\nu}_{\rm adj}$} & \textbf{$\hat{\nu}_{\rm Bayes}$} & \textbf{$\hat{\nu}_{\rm Pseudo}$} \\
\hline
1 & 1.4947 & 1.4934 & 1.5051 & 1.5080 \\
2 & 1.4923 & 1.4929 & 1.5029 & 1.5057 \\
5 & 1.4878 & 1.4916 & 1.4981 & 1.5004 \\
10 & 1.4880 & 1.4991 & 1.4969 & 1.4996 \\
20 & 1.4786 & 1.5074 & 1.4900 & 1.4911 \\
40 & 1.4730 & 1.5151 & 1.4803 & 1.4840 \\
60 & 1.4868 & 1.5491 & 1.4916 & 1.4955 \\
80 & 1.4624 & 1.5442 & 1.4696 & 1.4709 \\
\hline
\end{tabular*}
\end{table*}

\begin{table*}[h!]%
\centering %
\caption{\com{Bias of estimates of $\nu$}, with true $\nu = 2$ and sample size $n = 300$.\label{tab:bias_n300_nu2}}%
\begin{tabular*}{\textwidth}{@{\extracolsep\fill}rrrrr@{\extracolsep\fill}}
\hline
\textbf{$p$} & \textbf{$\hat{\nu}$} & \textbf{$\hat{\nu}_{\rm adj}$} & \textbf{$\hat{\nu}_{\rm Bayes}$} & \textbf{$\hat{\nu}_{\rm Pseudo}$} \\
\hline
1 & 0.0758 & 0.0883 & 0.0921 & 0.0984 \\
2 & 0.0686 & 0.0884 & 0.0849 & 0.0912 \\
5 & 0.0428 & 0.0852 & 0.0588 & 0.0654 \\
10 & 0.0102 & 0.0929 & 0.0260 & 0.0328 \\
20 & $-0.0700$ & 0.1061 & $-0.0550$ & $-0.0474$ \\
40 & $-0.3642$ & 0.1628 & $-0.3445$ & $-0.3214$ \\
60 & $-1.3767$ & 0.3147 & $-1.3986$ & $-1.3538$ \\
80 & $-1.6754$ & 0.6271 & $-1.6715$ & $-1.6721$ \\
\hline
\end{tabular*}
\end{table*}

\begin{table*}[h!]%
\centering %
\caption{\com{Bias of estimates of $\nu$}, with true $\nu = 5$ and sample size $n = 2500$.\label{tab:bias_n2500_nu5}}%
\begin{tabular*}{\textwidth}{@{\extracolsep\fill}rrrrr@{\extracolsep\fill}}
\hline
\textbf{$p$} & \textbf{$\hat{\nu}$} & \textbf{$\hat{\nu}_{\rm adj}$} & \textbf{$\hat{\nu}_{\rm Bayes}$} & \textbf{$\hat{\nu}_{\rm Pseudo}$} \\
\hline
1 & 0.0605 & 0.0638 & 0.0709 & 0.0738 \\
2 & 0.0572 & 0.0631 & 0.0676 & 0.0706 \\
5 & 0.0482 & 0.0616 & 0.0585 & 0.0615 \\
10 & 0.0347 & 0.0612 & 0.0451 & 0.0480 \\
20 & 0.0095 & 0.0621 & 0.0199 & 0.0228 \\
40 & $-0.0410$ & 0.0643 & $-0.0308$ & $-0.0277$ \\
60 & $-0.0935$ & 0.0671 & $-0.0834$ & $-0.0804$ \\
80 & $-0.1539$ & 0.0636 & $-0.1439$ & $-0.1409$ \\
\hline
\end{tabular*}
\end{table*}

\begin{table*}[h!]%
\centering %
\caption{\com{Bias of estimates of $\nu$}, with true $\nu = 10$ and sample size $n = 4500$.\label{tab:bias_n4500_nu10}}%
\begin{tabular*}{\textwidth}{@{\extracolsep\fill}rrrrr@{\extracolsep\fill}}
\hline
\textbf{$p$} & \textbf{$\hat{\nu}$} & \textbf{$\hat{\nu}_{\rm adj}$} & \textbf{$\hat{\nu}_{\rm Bayes}$} & \textbf{$\hat{\nu}_{\rm Pseudo}$} \\
\hline
1 & 0.2310 & 0.2319 & 0.2540 & 0.2603 \\
2 & 0.2264 & 0.2313 & 0.2497 & 0.2558 \\
5 & 0.2153 & 0.2304 & 0.2389 & 0.2446 \\
10 & 0.1993 & 0.2329 & 0.2224 & 0.2281 \\
20 & 0.1678 & 0.2405 & 0.1907 & 0.1969 \\
40 & 0.0990 & 0.2398 & 0.1213 & 0.1280 \\
60 & 0.0393 & 0.2602 & 0.0609 & 0.0678 \\
80 & $-0.0347$ & 0.2525 & $-0.0112$ & $-0.0052$ \\
\hline
\end{tabular*}
\end{table*}

\begin{table*}[h!]%
\centering %
\caption{\com{Standard error of estimates of $\nu$}, with true $\nu = 2$ and sample size $n = 300$.\label{tab:se_n300_nu2}}%
\begin{tabular*}{\textwidth}{@{\extracolsep\fill}rrrrr@{\extracolsep\fill}}
\hline
\textbf{$p$} & \textbf{$\hat{\nu}$} & \textbf{$\hat{\nu}_{\rm adj}$} & \textbf{$\hat{\nu}_{\rm Bayes}$} & \textbf{$\hat{\nu}_{\rm Pseudo}$} \\
\hline
1 & 0.3373 & 0.3396 & 0.3423 & 0.3438 \\
2 & 0.3373 & 0.3412 & 0.3425 & 0.3441 \\
5 & 0.3337 & 0.3428 & 0.3389 & 0.3406 \\
10 & 0.3362 & 0.3535 & 0.3414 & 0.3430 \\
20 & 0.3311 & 0.3699 & 0.3363 & 0.3380 \\
40 & 0.4766 & 0.4193 & 0.4713 & 0.4533 \\
60 & 0.4412 & 0.4922 & 0.4250 & 0.4847 \\
80 & 0.1250 & 0.6536 & 0.1186 & 0.1271 \\
\hline
\end{tabular*}
\end{table*}

\begin{table*}[h!]%
\centering %
\caption{\com{Standard error of estimates of $\nu$}, with true $\nu = 5$ and sample size $n = 2500$.\label{tab:se_n2500_nu5}}%
\begin{tabular*}{\textwidth}{@{\extracolsep\fill}rrrrr@{\extracolsep\fill}}
\hline
\textbf{$p$} & \textbf{$\hat{\nu}$} & \textbf{$\hat{\nu}_{\rm adj}$} & \textbf{$\hat{\nu}_{\rm Bayes}$} & \textbf{$\hat{\nu}_{\rm Pseudo}$} \\
\hline
1 & 0.5378 & 0.5380 & 0.5400 & 0.5405 \\
2 & 0.5379 & 0.5385 & 0.5401 & 0.5407 \\
5 & 0.5372 & 0.5388 & 0.5394 & 0.5400 \\
10 & 0.5385 & 0.5421 & 0.5407 & 0.5413 \\
20 & 0.5364 & 0.5434 & 0.5386 & 0.5392 \\
40 & 0.5362 & 0.5506 & 0.5384 & 0.5389 \\
60 & 0.5299 & 0.5536 & 0.5321 & 0.5327 \\
80 & 0.5214 & 0.5524 & 0.5236 & 0.5242 \\
\hline
\end{tabular*}
\end{table*}

\begin{table*}[h!]%
\centering %
\caption{\com{Standard error of estimates of $\nu$}, with true $\nu = 10$ and sample size $n = 4500$.\label{tab:se_n4500_nu10}}%
\begin{tabular*}{\textwidth}{@{\extracolsep\fill}rrrrr@{\extracolsep\fill}}
\hline
\textbf{$p$} & \textbf{$\hat{\nu}$} & \textbf{$\hat{\nu}_{\rm adj}$} & \textbf{$\hat{\nu}_{\rm Bayes}$} & \textbf{$\hat{\nu}_{\rm Pseudo}$} \\
\hline
1 & 1.4782 & 1.4768 & 1.4850 & 1.4868 \\
2 & 1.4765 & 1.4764 & 1.4835 & 1.4853 \\
5 & 1.4736 & 1.4751 & 1.4804 & 1.4818 \\
10 & 1.4761 & 1.4824 & 1.4817 & 1.4836 \\
20 & 1.4705 & 1.4895 & 1.4792 & 1.4795 \\
40 & 1.4711 & 1.4975 & 1.4768 & 1.4799 \\
60 & 1.4878 & 1.5287 & 1.4919 & 1.4954 \\
80 & 1.4635 & 1.5250 & 1.4711 & 1.4724 \\
\hline
\end{tabular*}
\end{table*}

\begin{table*}[h!]%
\centering %
\caption{\com{Median numerical second derivatives of the objective functions with respect to $\omega$, evaluated at the returned estimates, with true $\nu = 2$ and sample size $n = 300$.\label{tab:curvature_n300_nu2}}}%
\begin{tabular*}{\textwidth}{@{\extracolsep\fill}rrrrr@{\extracolsep\fill}}
\hline
\textbf{$p$} & \textbf{$\hat{\nu}$} & \textbf{$\hat{\nu}_{\rm adj}$} & \textbf{$\hat{\nu}_{\rm Bayes}$} & \textbf{$\hat{\nu}_{\rm Pseudo}$} \\
\hline
1 & $1.8631 \times 10^2$ & $1.8727 \times 10^2$ & $1.8719\times 10^2$ & $1.8842\times 10^2$ \\
2 & $1.8347 \times 10^2$ & $1.8735\times 10^2$ & $1.8684 \times 10^2$ & $1.8717\times 10^2$ \\
5 & $1.7706 \times 10^2$ & $1.8386 \times 10^2$& $1.8086 \times 10^2$ & $1.8143 \times 10^2$ \\
10 & $1.7128 \times 10^2$ & $1.8626 \times 10^2 $& $1.7369 \times 10^2$ & $1.7474 \times 10^2 $\\
20 & $1.5718 \times 10^2$ & $2.0002\times 10^2 $& $1.5951 \times 10^2$ & $1.5961 \times 10^2$ \\
40 & $1.3344 \times 10^2$ & $ 1.8532 \times 10^4$& $1.3424 \times 10^2$ & $1.4097 \times 10^2 $\\
60 & $8.5592 \times 10^8$ & $3.1292 \times 10^4$ & $9.0068 \times 10^8$ & $7.5867 \times 10^8$\\
80 & $6.2513\times 10^9$ & $1.2608 \times 10^5$& $6.1464 \times 10^9$& $6.2481 \times 10^9 $\\
\hline
\end{tabular*}
\end{table*}

\begin{table*}[h!]%
\centering %
\caption{\com{Median numerical second derivatives of the objective functions with respect to $\omega$, evaluated at the returned estimates, with true $\nu = 5$ and sample size $n = 2500$.\label{tab:curvature_n2500_nu5}}}%
\begin{tabular*}{\textwidth}{@{\extracolsep\fill}rrrrr@{\extracolsep\fill}}
\hline
\textbf{$p$} & \textbf{$\hat{\nu}$} & \textbf{$\hat{\nu}_{\rm adj}$} & \textbf{$\hat{\nu}_{\rm Bayes}$} & \textbf{$\hat{\nu}_{\rm Pseudo}$} \\
\hline
1 & $2.5363 \times 10^3$ & $2.5381 \times 10^3$ & $2.5430 \times 10^3$ & $2.5435 \times 10^3$ \\
2 & $2.5324 \times 10^3$ & $2.5424 \times 10^3$ & $2.5390 \times 10^3$ & $2.5395 \times 10^3$ \\
5 & $2.5318 \times 10^3$ & $2.5378 \times 10^3$ & $2.5391 \times 10^3$ & $2.5357 \times 10^3$ \\
10 & $2.5182 \times 10^3$ & $2.5627 \times 10^3$ & $2.5271 \times 10^3$ & $2.5239 \times 10^3$ \\
20 & $2.5125 \times 10^3$ & $2.6373 \times 10^3$ & $2.5200 \times 10^3$ & $2.5040 \times 10^3$ \\
40 & $2.4622 \times 10^3$ & $2.5362 \times 10^3$ & $2.4729 \times 10^3$ & $2.4720 \times 10^3$ \\
60 & $2.4095 \times 10^3$ & $2.4791 \times 10^3$ & $2.4156 \times 10^3$ & $2.4162 \times 10^3$ \\
80 & $2.3516 \times 10^3$ & $2.4442 \times 10^3$ & $2.3635 \times 10^3$ & $2.3621 \times 10^3$ \\
\hline
\end{tabular*}
\end{table*}

\begin{table*}[h!]%
\centering %
\caption{\com{Median numerical second derivatives of the objective functions with respect to $\omega$, evaluated at the returned estimates, with true $\nu = 10$ and sample size $n = 4500$.\label{tab:curvature_n4500_nu10}}}%
\begin{tabular*}{\textwidth}{@{\extracolsep\fill}rrrrr@{\extracolsep\fill}}
\hline
\textbf{$p$} & \textbf{$\hat{\nu}$} & \textbf{$\hat{\nu}_{\rm adj}$} & \textbf{$\hat{\nu}_{\rm Bayes}$} & \textbf{$\hat{\nu}_{\rm Pseudo}$} \\
\hline
1 & $5.5132 \times 10^3$ & $5.5191 \times 10^3$ & $5.5306 \times 10^3$ & $5.5296 \times 10^3$ \\
2 & $5.5006 \times 10^3$ & $5.5029 \times 10^3$ & $5.5101 \times 10^3$ & $5.5104 \times 10^3$ \\
5 & $5.4845 \times 10^3$ & $5.4876 \times 10^3$ & $5.4935 \times 10^3$ & $5.4930 \times 10^3$ \\
10 & $5.5832 \times 10^3$ & $5.7017 \times 10^3$ & $5.5699 \times 10^3$ & $5.5774 \times 10^3$ \\
20 & $5.6044 \times 10^3$ & $6.2756 \times 10^3$ & $5.6556 \times 10^3$ & $5.6404 \times 10^3$ \\
40 & $5.4683 \times 10^3$ & $5.8298 \times 10^3$ & $5.4744 \times 10^3$ & $5.4744 \times 10^3$ \\
60 & $5.4818 \times 10^3$ & $1.0989 \times 10^4$ & $5.4821 \times 10^3$ & $5.4614 \times 10^3$ \\
80 & $5.2934 \times 10^3$ & $5.4792 \times 10^3$ & $5.3283 \times 10^3$ & $5.2910 \times 10^3$ \\
\hline
\end{tabular*}
\end{table*}

\begin{table*}[h!]%
\centering %
\caption{\com{Median numerical first derivatives of the objective functions with respect to $\omega$, evaluated at the returned estimates, with true $\nu = 2$ and sample size $n = 300$.\label{tab:derivative_n300_nu2}}}%
\begin{tabular*}{\textwidth}{@{\extracolsep\fill}rrrrr@{\extracolsep\fill}}
\hline
\textbf{$p$} & \textbf{$\hat{\nu}$} & \textbf{$\hat{\nu}_{\rm adj}$} & \textbf{$\hat{\nu}_{\rm Bayes}$} & \textbf{$\hat{\nu}_{\rm Pseudo}$} \\
\hline
1 & $7.6000 \times 10^{-5}$ & $1.8400 \times 10^{-4}$ & $7.6000 \times 10^{-5}$ & $7.6000 \times 10^{-5}$ \\
2 & $6.9000 \times 10^{-5}$ & $2.4300 \times 10^{-4}$ & $7.3000 \times 10^{-5}$ & $7.4000 \times 10^{-5}$ \\
5 & $6.3000 \times 10^{-5}$ & $2.5700 \times 10^{-4}$ & $6.6000 \times 10^{-5}$ & $6.9000 \times 10^{-5}$ \\
10 & $6.5000 \times 10^{-5}$ & $1.2100 \times 10^{-3}$ & $6.6000 \times 10^{-5}$ & $6.5000 \times 10^{-5}$ \\
20 & $9.5000 \times 10^{-5}$ & $2.0790 \times 10^{-2}$ & $8.9000 \times 10^{-5}$ & $8.8000 \times 10^{-5}$ \\
40 & $1.9580 \times 10^{-3}$ & $4.8218 \times 10^{-1}$ & $1.5780 \times 10^{-3}$ & $1.7280 \times 10^{-3}$ \\
60 & $8.9492 \times 10^4$ & $9.3615 \times 10^{-1}$ & $1.1853 \times 10^5$ & $1.0607 \times 10^5$ \\
80 & $8.0224 \times 10^5$ & $2.3765 \times 10^0$ & $8.2470 \times 10^5$ & $7.8945 \times 10^5$ \\
\hline
\end{tabular*}
\end{table*}

\begin{table*}[h!]%
\centering %
\caption{\com{Median numerical first derivatives of the objective functions with respect to $\omega$, evaluated at the returned estimates, with true $\nu = 5$ and sample size $n = 2500$.\label{tab:derivative_n2500_nu5}}}%
\begin{tabular*}{\textwidth}{@{\extracolsep\fill}rrrrr@{\extracolsep\fill}}
\hline
\textbf{$p$} & \textbf{$\hat{\nu}$} & \textbf{$\hat{\nu}_{\rm adj}$} & \textbf{$\hat{\nu}_{\rm Bayes}$} & \textbf{$\hat{\nu}_{\rm Pseudo}$} \\
\hline
1 & $2.5930 \times 10^{-3}$ & $2.6020 \times 10^{-3}$ & $2.5910 \times 10^{-3}$ & $2.5930 \times 10^{-3}$ \\
2 & $2.5880 \times 10^{-3}$ & $2.6020 \times 10^{-3}$ & $2.5990 \times 10^{-3}$ & $2.5960 \times 10^{-3}$ \\
5 & $2.5820 \times 10^{-3}$ & $2.7450 \times 10^{-3}$ & $2.5840 \times 10^{-3}$ & $2.5880 \times 10^{-3}$ \\
10 & $2.5900 \times 10^{-3}$ & $3.1690 \times 10^{-3}$ & $2.6200 \times 10^{-3}$ & $2.5960 \times 10^{-3}$ \\
20 & $2.6700 \times 10^{-3}$ & $9.0205 \times 10^{-2}$ & $2.7420 \times 10^{-3}$ & $2.6860 \times 10^{-3}$ \\
40 & $2.4260 \times 10^{-3}$ & $3.5837 \times 10^{-2}$ & $2.5200 \times 10^{-3}$ & $2.5020 \times 10^{-3}$ \\
60 & $2.3910 \times 10^{-3}$ & $5.2430 \times 10^{-3}$ & $2.3780 \times 10^{-3}$ & $2.3710 \times 10^{-3}$ \\
80 & $2.3310 \times 10^{-3}$ & $3.3480 \times 10^{-3}$ & $2.3570 \times 10^{-3}$ & $2.3480 \times 10^{-3}$ \\
\hline
\end{tabular*}
\end{table*}

\begin{table*}[h!]%
\centering %
\caption{\com{Median numerical first derivatives of the objective functions with respect to $\omega$, evaluated at the returned estimates, with true $\nu = 10$ and sample size $n = 4500$.\label{tab:derivative_n4500_nu10}}}%
\begin{tabular*}{\textwidth}{@{\extracolsep\fill}rrrrr@{\extracolsep\fill}}
\hline
\textbf{$p$} & \textbf{$\hat{\nu}$} & \textbf{$\hat{\nu}_{\rm adj}$} & \textbf{$\hat{\nu}_{\rm Bayes}$} & \textbf{$\hat{\nu}_{\rm Pseudo}$} \\
\hline
1 & $2.8620 \times 10^{-3}$ & $3.4720 \times 10^{-3}$ & $2.8720 \times 10^{-3}$ & $2.8210 \times 10^{-3}$ \\
2 & $2.5070 \times 10^{-3}$ & $2.7630 \times 10^{-3}$ & $2.5440 \times 10^{-3}$ & $2.5770 \times 10^{-3}$ \\
5 & $2.5430 \times 10^{-3}$ & $2.9430 \times 10^{-3}$ & $2.5480 \times 10^{-3}$ & $2.5260 \times 10^{-3}$ \\
10 & $8.1740 \times 10^{-3}$ & $5.5372 \times 10^{-2}$ & $6.9670 \times 10^{-3}$ & $6.8920 \times 10^{-3}$ \\
20 & $2.9573 \times 10^{-2}$ & $8.6169 \times 10^{-1}$ & $2.9893 \times 10^{-2}$ & $3.9329 \times 10^{-2}$ \\
40 & $1.1431 \times 10^{-2}$ & $5.3797 \times 10^{-1}$ & $1.2235 \times 10^{-2}$ & $1.1591 \times 10^{-2}$ \\
60 & $1.4962 \times 10^{-2}$ & $2.3084 \times 10^0$ & $1.6225 \times 10^{-2}$ & $1.8099 \times 10^{-2}$ \\
80 & $4.7970 \times 10^{-3}$ & $3.7885 \times 10^{-1}$ & $4.8680 \times 10^{-3}$ & $4.7760 \times 10^{-3}$ \\
\hline
\end{tabular*}
\end{table*}

\begin{figure}[h!] \centering 
    \begin{subfigure}{0.48\textwidth} \centering \includegraphics[width=\linewidth]{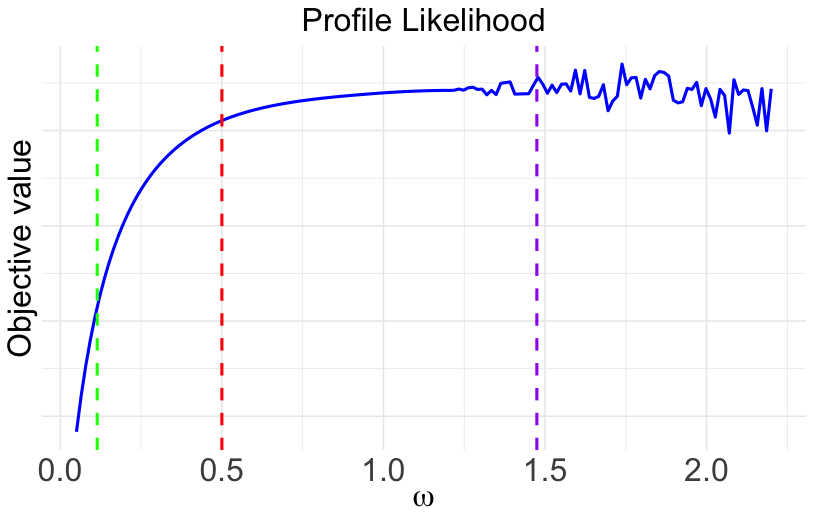} \caption{Objective function of Profile likelihood approach}
    \end{subfigure} 
    \hfill 
    \begin{subfigure}{0.48\textwidth} \centering \includegraphics[width=\linewidth]{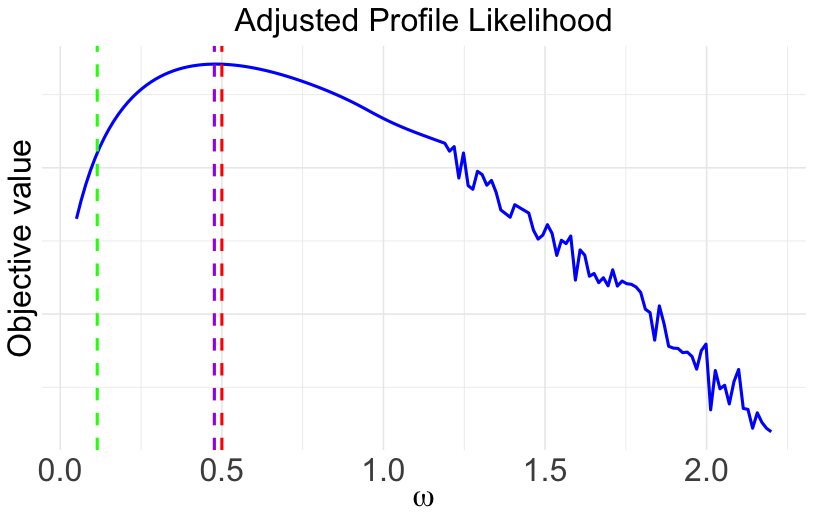} \caption{Objective function of Adjusted profile likelihood approach}
    \end{subfigure} 
    \vspace{0.5em} 
    \begin{subfigure}{0.48\textwidth} \centering \includegraphics[width=\linewidth]{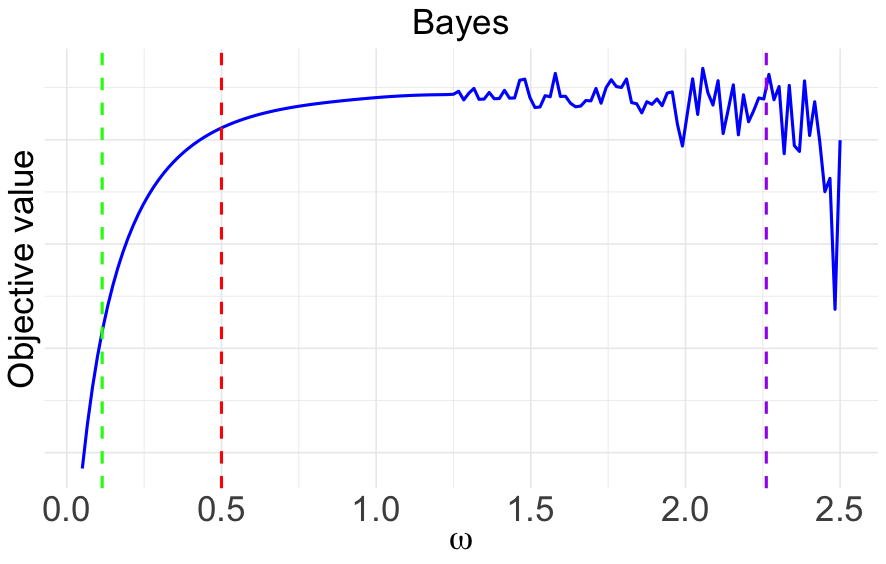} \caption{Objective function of Bayes approach} 
    \end{subfigure} 
    \hfill 
    \begin{subfigure}{0.48\textwidth} \centering \includegraphics[width=\linewidth]{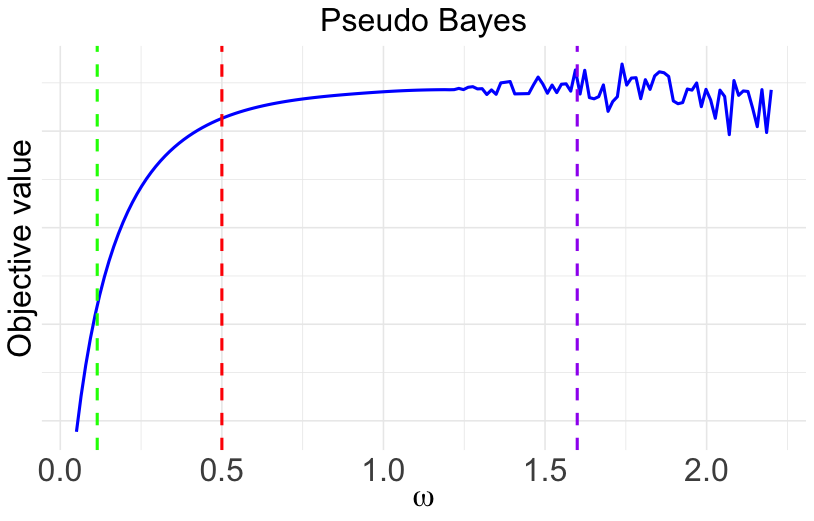} \caption{Objective function of Pseudo Bayes approach} 
    \end{subfigure} 
    \caption{\com{Objective functions for a simulated dataset with
    $n=300$, $p=60$, true $\nu=2$, and data-generation seed 2.
    The adjusted profile likelihood exhibits a clear interior peak,
    with the estimate returned by the optimizer close to the peak.
    In contrast, the objective functions for the profile likelihood
    and the two Bayesian MAP-based approaches exhibit numerically
    unstable behaviour near the estimates returned by the optimizer.
    The red lines indicate the true value $\omega=0.5$, the green lines indicate the data-dependent initial $\omega$ value, and the purple lines indicate the estimates returned by the optimizer for the corresponding approach.}}
\label{fig:derivative}
\end{figure}

\section{Additional simulation results for regression coefficient estimation}

\com{
We present additional results for the fully simulated setting with
$t_2$ errors and for the contamination experiments described in
Section~4 of the main text. The latter examine robustness to
contamination of the error distribution or two-point replacement
of the response values, with $p=3$ and $p=80$.
}

\com{
\subsection{Fully simulated setting with $t_2$ errors}

We report additional results for the third simulation setting described in Section~3.1 of the main manuscript. 
The design matrix was  generated from the $N_p(0, I)$ distribution, and we set $\beta = \mathbf{0}$ and $\sigma = 1$. $\epsilon_i$ was simulated from the $t_2$. The three settings were $(n,p)=(20,4)$, $(500,4)$, and $(100,50)$.
Figures~\ref{fig:rmse_t_error_a} and~\ref{fig:rmse_t_error_b} have $p/n = 0.20, 0.008$ respectively.  The smallest value of the average $L_2$ error is achieved by fixing $\nu$ at $2$ or $3$. Using $\hat{\nu}_{\rm adj}$ gives comparable results, slightly better than Huber regression, and comparable to the best fixed value. In Figure~\ref{fig:rmse_t_error_c}, with $p/n = 0.50$, the lowest average $L_2$ error is achieved at larger fixed $\nu = 6$. Among the $\nu$-estimation approaches, again using $\hat{\nu}_{\rm adj}$ gives the smallest $L_2$ error, but this value is relatively large compared to the $L_2$ error that can be achieved by traditional methods, especially Huber with smaller values and adaptive estimation of $k$.
}

\begin{figure}[h!]
    \centering
    \resizebox{0.9\textwidth}{!}{%
    \begin{minipage}{\textwidth}
        \begin{subfigure}[t]{\linewidth}
            \centering
            \includegraphics[width=\linewidth]{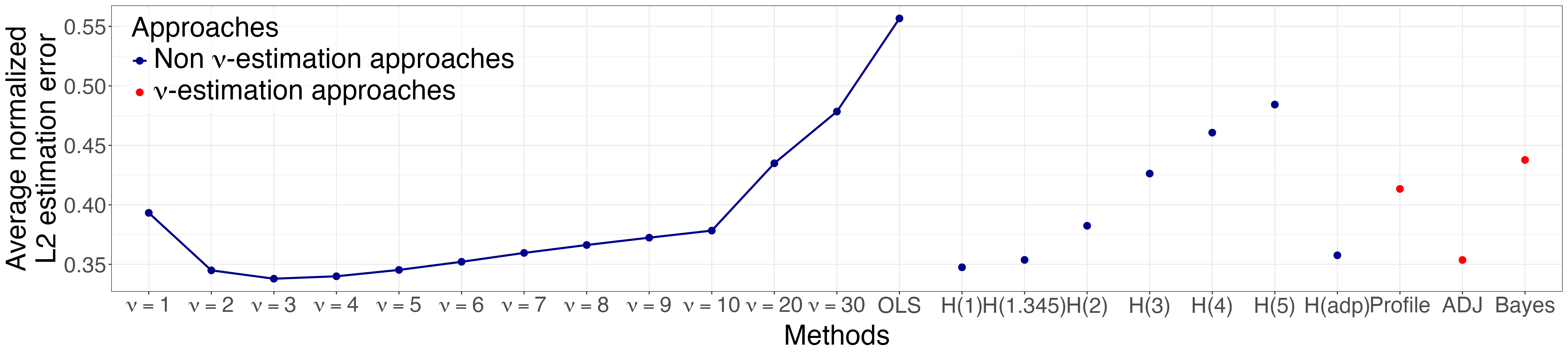}
            \caption{$n=20, p=4$}
            \label{fig:rmse_t_error_a}
        \end{subfigure}
        \vspace{1em}
        \begin{subfigure}[t]{\linewidth}
            \centering
            \includegraphics[width=\linewidth]{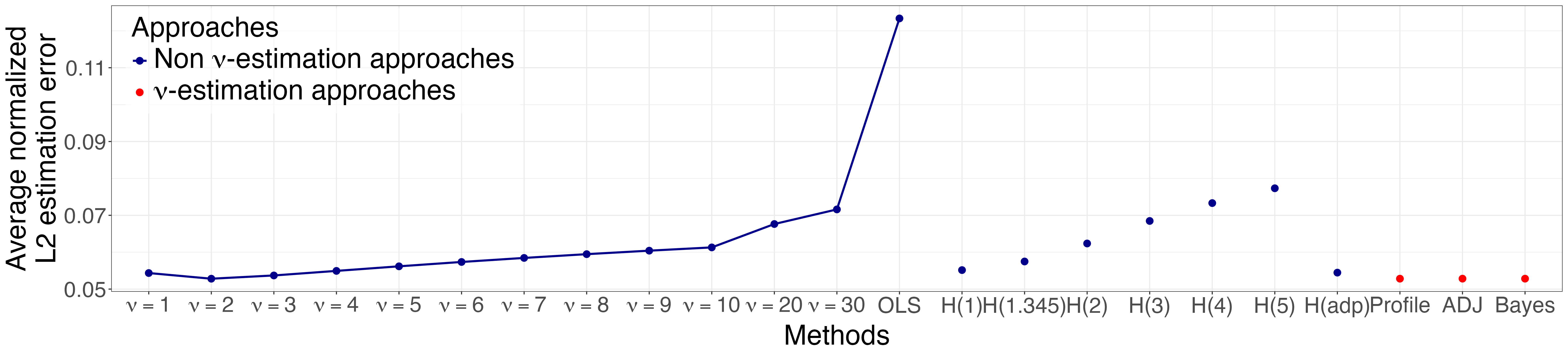}
            \caption{$n=500, p=4$}
            \label{fig:rmse_t_error_b}
        \end{subfigure}
        \vspace{1em}
        \begin{subfigure}[t]{\linewidth}
            \centering
            \includegraphics[width=\linewidth]{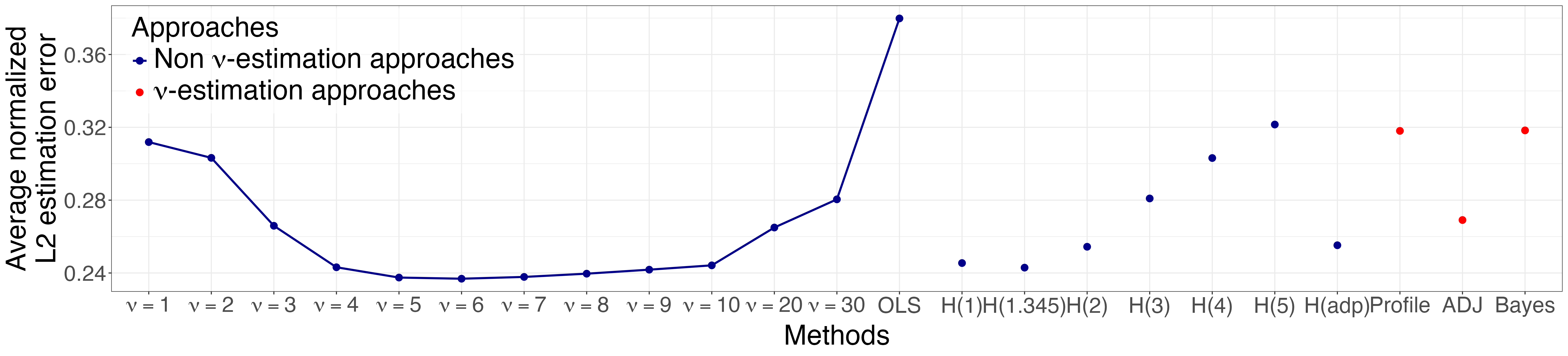}
            \caption{$n=100, p=50$}
            \label{fig:rmse_t_error_c}
        \end{subfigure}
    \end{minipage}
    }
    \caption{\com{Average $L_2$ error} of $\hat{\beta}$ with $t_2$ errors. Fixed values of $\nu$ and other OLS and Huber estimates are shown in blue. The Huber constant $k$ was either fixed ($H(k)$) or estimated ($H({\rm adp})$). Estimates of $\nu$, either by maximum likelihood, $\hat\nu$, adjusted profile likelihood, $\hat\nu_{\rm adj}$, or MAP with Jeffreys prior, $\hat\nu_{\rm Bayes}$, are shown in red. 
   }
    \label{fig:rmse_t_error}
\end{figure}

\subsection{Contamination settings with $p=3$}

\com{For $n=100$ and $p=3$, the maximum likelihood, adjusted
profile likelihood, and Bayesian MAP approaches to estimating $\nu$ give similar average $L_2$ errors for estimating $\beta$ under two-point contamination (Figure~\ref{fig:2pt_10_20_30}). Errors from estimating $\nu$ are also close to the best fixed-$\nu$ results under $t_2$ contamination (Figure~\ref{fig:t2_10_20_30}).
Their errors are close to the smallest errors among the fixed-$\nu$ values examined and are lower than those of OLS and adaptive Huber regression. 
The best-performing fixed value shifts towards smaller $\nu$ as the contamination rate increases. Huber regression is sensitive to the tuning constant $k$, and adaptive selection does not consistently give smaller errors than the fixed-$k$ choices examined.}

\subsection{Contamination settings with $p=80$}

\com{For $n=300$ and $p=80$, estimating $\nu$ by adjusted profile likelihood gives smaller average $L_2$ errors for estimating $\beta$ than using the maximum likelihood or Bayesian MAP estimate of $\nu$ under $\chi^2_4-4$ and $t_2$ contamination (Figures~\ref{fig:chi4_10_20_30_highdim} and~\ref{fig:t2_10_20_30_highdim}). 
Across all three contamination rates, the adjusted profile likelihood approach gives errors close to those of the best-performing fixed-$\nu$ and fixed-$k$ Huber fits examined, and lower than those of OLS and adaptive Huber regression.}

\begin{figure}[h!]
    \centering
    \resizebox{0.9\textwidth}{!}{%
    \begin{minipage}{\textwidth}

        \begin{subfigure}[t]{\linewidth}
            \centering
            \includegraphics[width=\linewidth]{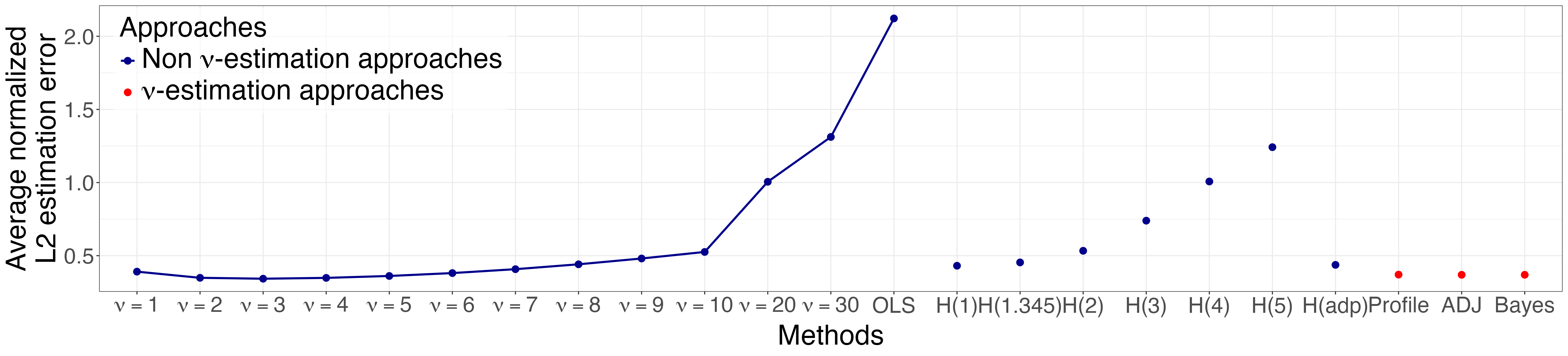}
            \caption{$10\%$ contamination}
        \end{subfigure}

        \vspace{1em}

        \begin{subfigure}[t]{\linewidth}
            \centering
            \includegraphics[width=\linewidth]{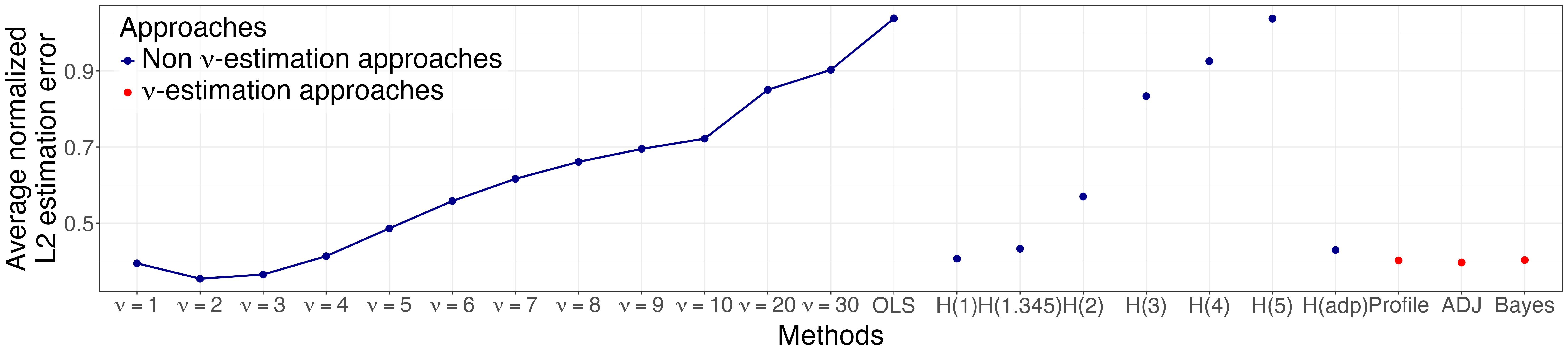}
            \caption{$20\%$ contamination}
        \end{subfigure}

        \vspace{1em}

        \begin{subfigure}[t]{\linewidth}
            \centering
            \includegraphics[width=\linewidth]{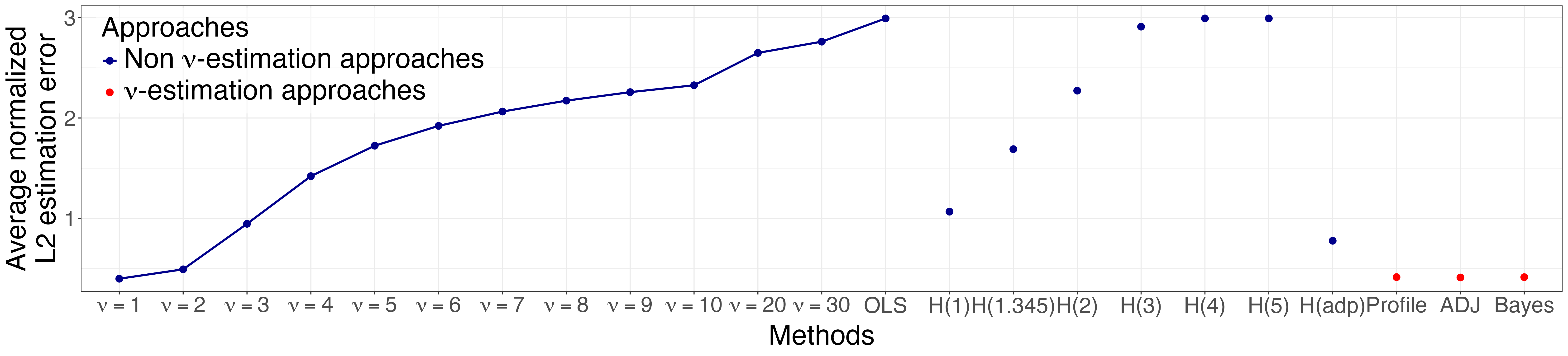}
            \caption{$30\%$ contamination}
        \end{subfigure}

    \end{minipage}
    
    }
    
    \caption{\com{Average $L_2$ error} of $\hat{\beta}$ with 10\%, 20\%, and 30\% two-point contaminated errors when $p = 3$. The plot compares non-$\nu$-estimation approaches (in blue) across various fixed values of $\nu$ and other traditional estimators (OLS and Huber) with the $\nu$-estimation approaches considered (in red). Here, $\nu = c$ denotes the approach in which $\nu$ is fixed at the constant value $c$; OLS denotes the ordinary least squares estimator; $H(c)$ denotes Huber regression with tuning parameter $k = c$; and $H(\rm adp)$ denotes adaptive Huber regression.}
    \label{fig:2pt_10_20_30}
\end{figure}

\newpage 

\begin{figure}[h!]
    \centering

    \resizebox{0.9\textwidth}{!}{%
    \begin{minipage}{\textwidth}

        \begin{subfigure}[t]{\linewidth}
            \centering
            \includegraphics[width=\linewidth]{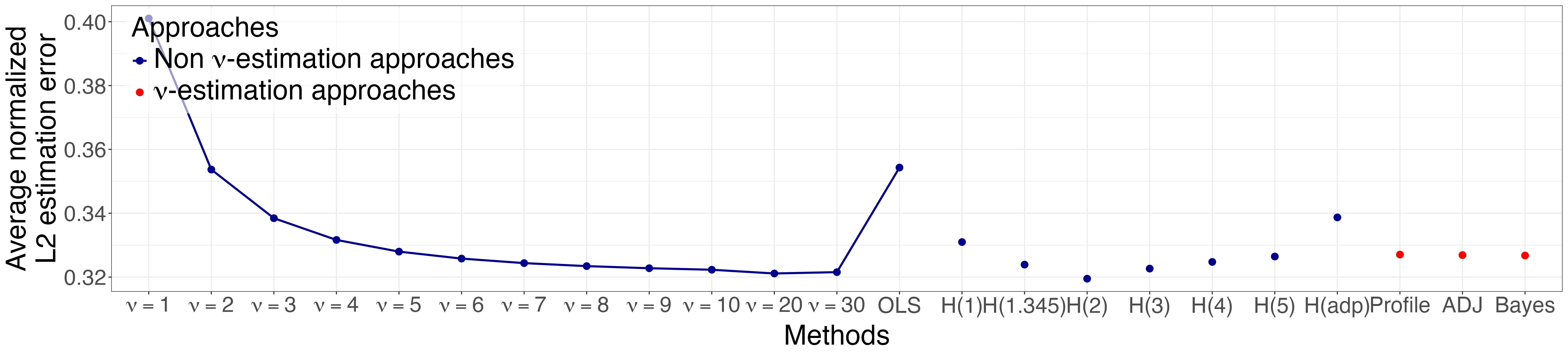 }
            \caption{$10\%$ contamination}
        \end{subfigure}

        \vspace{1em}

        \begin{subfigure}[t]{\linewidth}
            \centering
            \includegraphics[width=\linewidth]{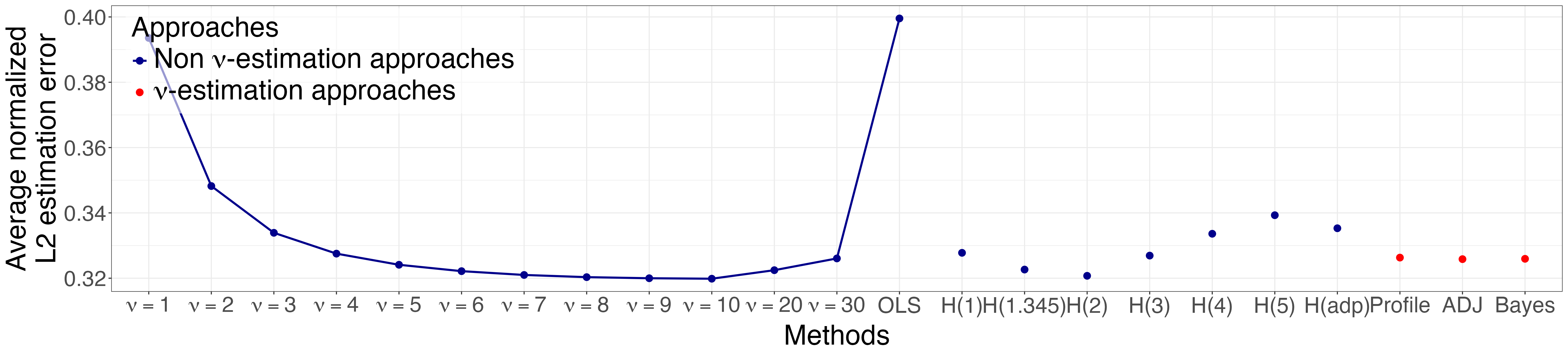}
            \caption{$20\%$ contamination}
        \end{subfigure}

        \vspace{1em}

        \begin{subfigure}[t]{\linewidth}
            \centering
            \includegraphics[width=\linewidth]{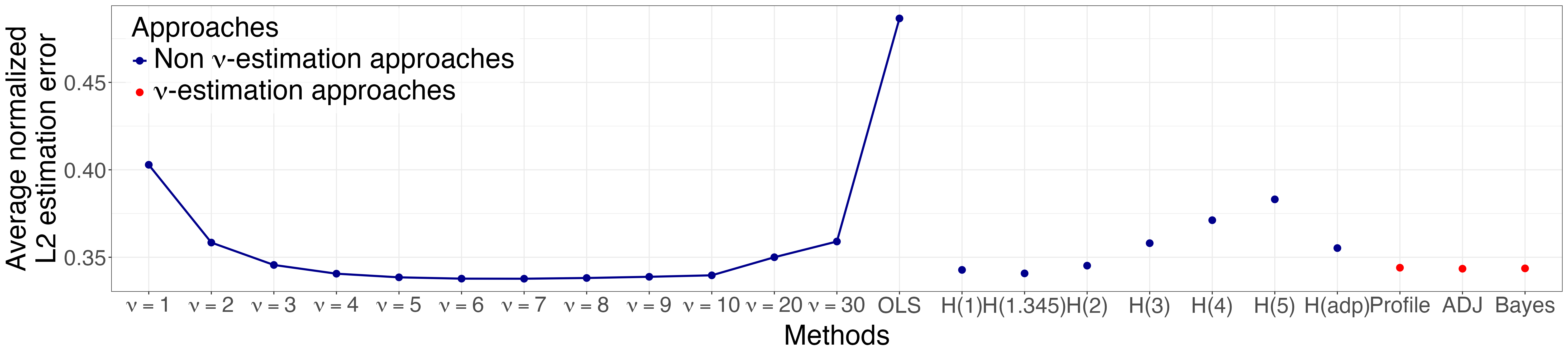}
            \caption{$30\%$ contamination}
        \end{subfigure}

    \end{minipage}
    }
    
    \caption{\com{Average $L_2$ error} of $\hat{\beta}$ with 10\%, 20\%, and 30\% $t_2$ contaminated errors when $p = 3$. The plot compares non-$\nu$-estimation approaches (in blue) across various fixed values of $\nu$ and other traditional estimators (OLS and Huber) with the $\nu$-estimation approaches considered (in red). Here, $\nu = c$ denotes the approach in which $\nu$ is fixed at the constant value $c$; OLS denotes the ordinary least squares estimator; $H(c)$ denotes Huber regression with tuning parameter $k = c$; and $H(\rm adp)$ denotes adaptive Huber regression.}
    
    \label{fig:t2_10_20_30}
\end{figure}

\begin{figure}[h!]
    \centering

    \resizebox{0.9\textwidth}{!}{%
    \begin{minipage}{\textwidth}

        \begin{subfigure}[t]{\linewidth}
            \centering
            \includegraphics[width=\linewidth]{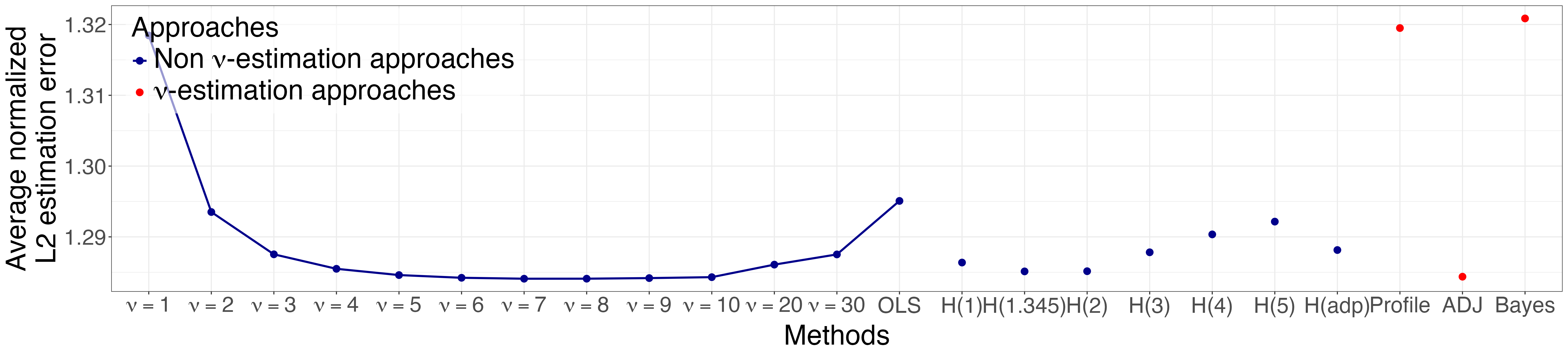}
            \caption{$10\%$ contamination}
        \end{subfigure}

        \vspace{1em}

        \begin{subfigure}[t]{\linewidth}
            \centering
            \includegraphics[width=\linewidth]{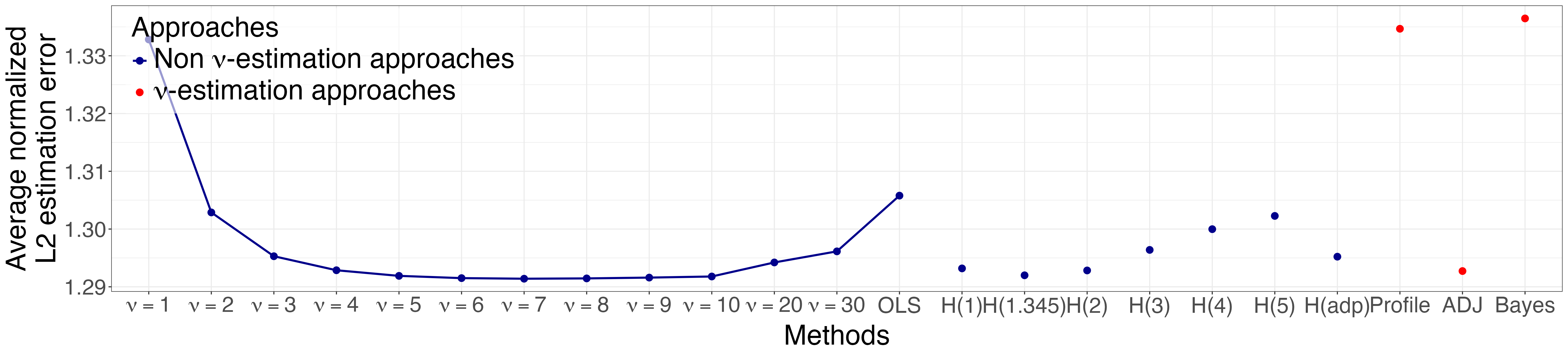}
            \caption{$20\%$ contamination}
        \end{subfigure}

        \vspace{1em}

        \begin{subfigure}[t]{\linewidth}
            \centering
            \includegraphics[width=\linewidth]{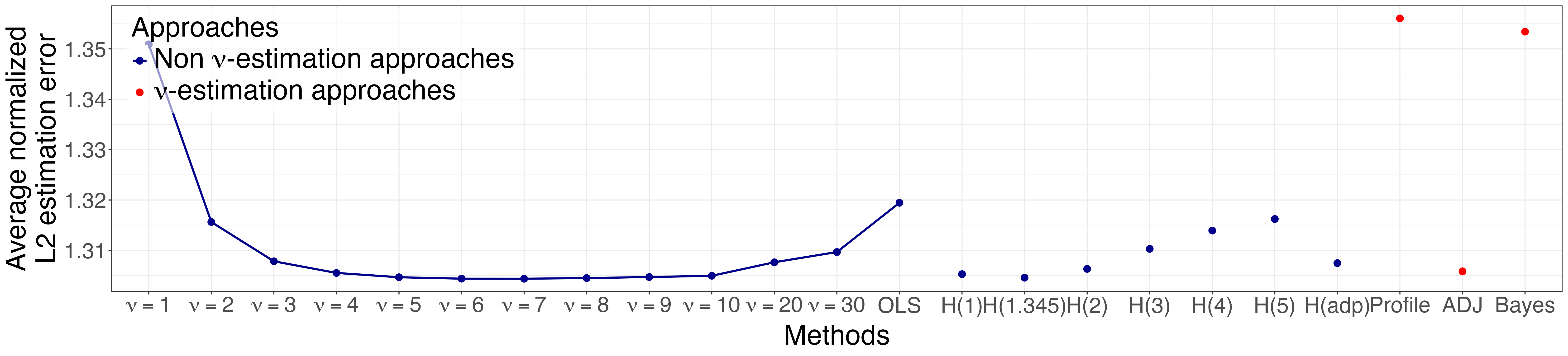}
            \caption{$30\%$ contamination}
        \end{subfigure}

    \end{minipage}
    }
    \caption{\com{Average $L_2$ error} of $\hat{\beta}$ with 10\%, 20\%, and 30\% $\chi^2_4 - 4$ contaminated errors when $p = 80$. The plot compares non-$\nu$-estimation approaches (in blue) across various fixed values of $\nu$ and other traditional estimators (OLS and Huber) with the $\nu$-estimation approaches considered (in red). Here, $\nu = c$ denotes the approach in which $\nu$ is fixed at the constant value $c$; OLS denotes the ordinary least squares estimator; $H(c)$ denotes Huber regression with tuning parameter $k = c$; and $H(\rm adp)$ denotes adaptive Huber regression.
    }
     \label{fig:chi4_10_20_30_highdim}
\end{figure}

\newpage

\begin{figure}[h!]
    \centering

    \resizebox{0.9\textwidth}{!}{%
    \begin{minipage}{\textwidth}

        \begin{subfigure}[t]{\linewidth}
            \centering
            \includegraphics[width=\linewidth]{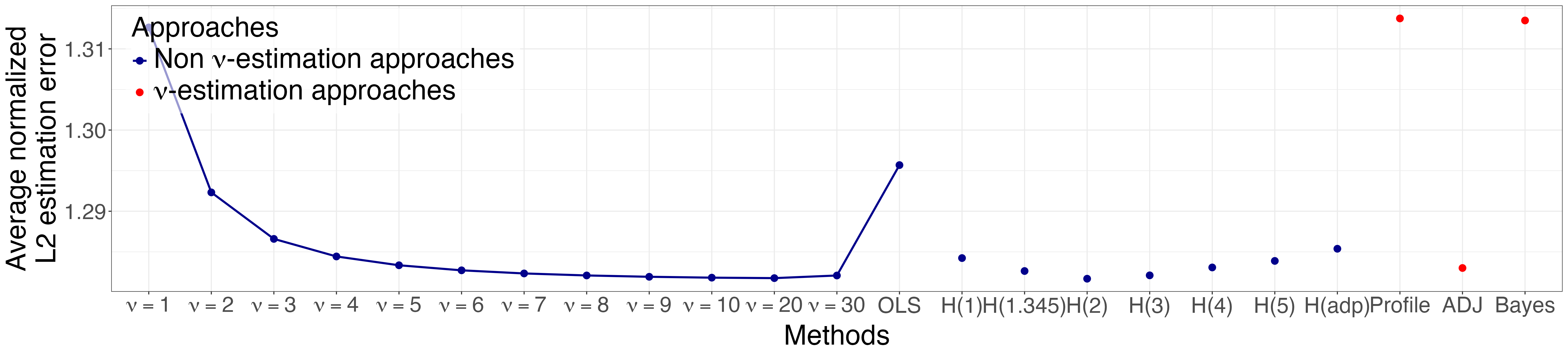}
            \caption{$10\%$ contamination}
        \end{subfigure}

        \vspace{1em}

        \begin{subfigure}[t]{\linewidth}
            \centering
            \includegraphics[width=\linewidth]{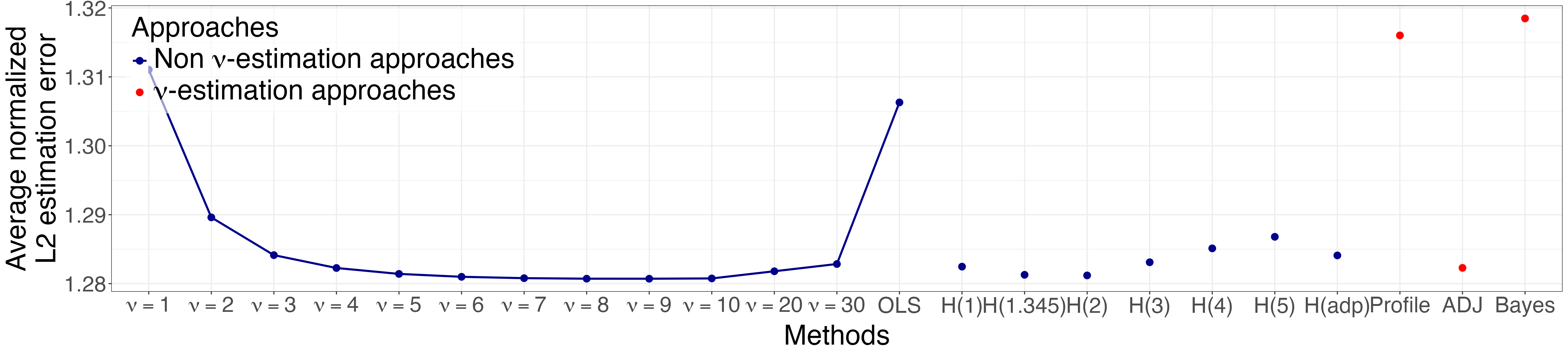}
            \caption{$20\%$ contamination}
        \end{subfigure}

        \vspace{1em}

        \begin{subfigure}[t]{\linewidth}
            \centering
            \includegraphics[width=\linewidth]{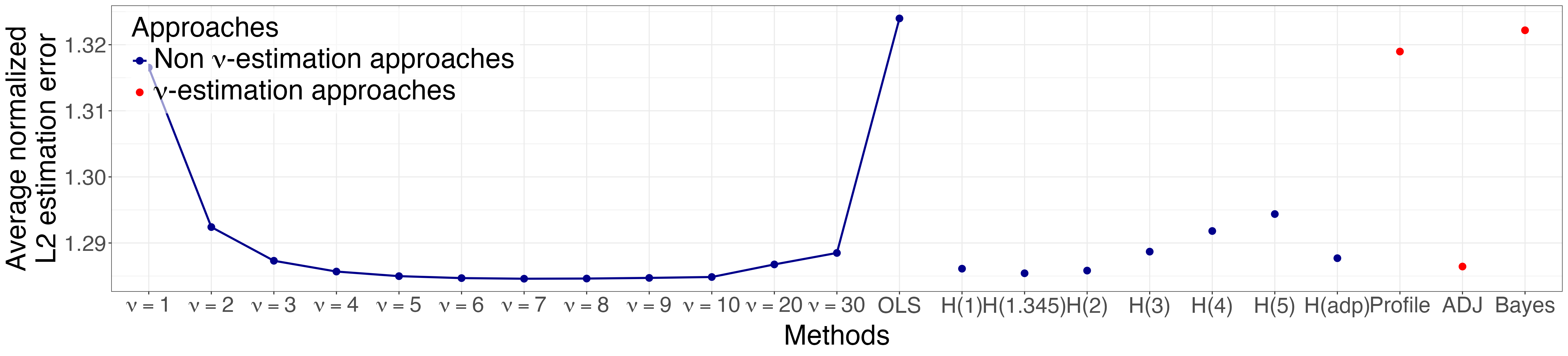}
            \caption{$30\%$ contamination}
        \end{subfigure}

    \end{minipage}
    }
    \caption{\com{Average $L_2$ error} of $\hat{\beta}$ with 10\%, 20\%, and 30\% $t_2$ contaminated errors when $p = 80$. The plot compares non-$\nu$-estimation approaches (in blue) across various fixed values of $\nu$ and other traditional estimators (OLS and Huber) with the $\nu$-estimation approaches considered (in red). Here, $\nu = c$ denotes the approach in which $\nu$ is fixed at the constant value $c$; OLS denotes the ordinary least squares estimator; $H(c)$ denotes Huber regression with tuning parameter $k = c$; and $H(\rm adp)$ denotes adaptive Huber regression.}
    
    \label{fig:t2_10_20_30_highdim}
\end{figure}

\end{document}